\documentclass[pra,onecolumn,preprintnumbers,superscriptaddress]{revtex4}
\usepackage[a4paper, left=0.8in, right=0.8in, top=0.8in, bottom=0.8in]{geometry}

\usepackage{graphicx}
\usepackage{bm}
\usepackage{amsmath,hyperref}
\usepackage{amssymb}
\usepackage{amsfonts}
\usepackage{fancyhdr}
\usepackage{slashed}
\usepackage{latexsym,bbm}
\usepackage{amsthm}
\usepackage{algorithm,algpseudocode}

\floatname{algorithm}{Procedure}

\newcommand{\Var}{\mathrm{Var}}

\newcommand{\bra}[1]{\langle{#1}|}
\newcommand{\ket}[1]{|{#1}\rangle}

\newcommand{\braket}[2]{\langle{#1}|{#2}\rangle}

\newcommand{\Tr}{\mathrm{Tr}}
\usepackage{color}
\definecolor{blue}{rgb}{0,0.2,1}

\definecolor{red}{rgb}{0.9,0,0}

\newcommand{\Ord}[1]{\mathcal{O}\left( #1 \right)}
\newcommand{\tOrd}[1]{\tilde {\mathcal{O}}\left( #1 \right)}

\newcommand{\Tht}[1]{\Theta \left( #1 \right)}

\theoremstyle{plain}
\newtheorem{prop}{Proposition}
\newtheorem{lemma}{Lemma}
\newtheorem{defn}{Definition}
\newtheorem{assume}{Assumption}

\def\be{\begin{eqnarray}}
\def\ee{\end{eqnarray}}

\usepackage{color}
\definecolor{Pr}{rgb}{0.4,0.3,0.9}

\usepackage{mathtools}
\DeclarePairedDelimiter{\ceil}{\lceil}{\rceil}

\newcommand{\norm}[1]{\left\lVert#1\right\rVert}
\newcommand{\realpar}[1]{{\rm Re}\left\{ #1\right\} }
\newcommand{\impar}[1]{{\rm Im}\left\{ #1\right\} }

\begin{document}

\title{Near-term quantum algorithms for linear systems of equations}

\author{Hsin-Yuan Huang }
\email{hsinyuan@caltech.edu}
\affiliation{Institute for Quantum Information and Matter, California Institute of Technology, USA}
\affiliation{Department of Computing and Mathematical Sciences, California Institute of Technology, USA}
\author{Kishor Bharti }
\email{kishor.bharti1@gmail.com}
\affiliation{Centre for Quantum Technologies, National University of Singapore, Singapore}
\author{Patrick Rebentrost }
\email{cqtfpr@nus.edu.sg}
\affiliation{Centre for Quantum Technologies, National University of Singapore, Singapore}
\date{\today}

\begin{abstract}
Solving linear systems of equations is essential for many problems in science and technology, including problems in machine learning.
Existing quantum algorithms have demonstrated the potential for large speedups, but the required quantum resources are not immediately available on near-term quantum devices.
In this work, we study near-term quantum algorithms for linear systems of equations of the form $Ax = b$.
We investigate the use of variational algorithms and analyze their optimization landscapes.
There exist types of linear systems for which variational algorithms designed to avoid barren plateaus, such as properly-initialized imaginary time evolution and adiabatic-inspired optimization,
suffer from a different plateau problem.
To circumvent this issue, we design near-term
algorithms based on a core idea: the classical combination of variational quantum states (CQS).
We exhibit several provable guarantees for these algorithms, supported by the representation of the linear system on a so-called Ansatz tree. The CQS approach and the Ansatz tree also admit the systematic application of heuristic approaches, including a gradient-based search. We have conducted numerical experiments solving linear systems as large as $2^{300} \times 2^{300}$ by considering cases where we can simulate the quantum algorithm efficiently on a classical computer.
These experiments demonstrate the algorithms' ability to scale to system sizes within reach in near-term quantum devices of about $100$-$300$ qubits.
\end{abstract}
\maketitle

\section{Introduction}

Quantum computing promises speedups for a set of problems including integer factoring and search. Speedups have also been discussed for finding approximate solutions to linear systems of equations and  convex optimization \cite{Harrow2009,Apeldoorn2018,Chakrabarti2018,vAG18,LCW19,Brandao2019}. Many of these algorithms require a large amount of low-decoherence and fully connected quantum bits, beyond the reach of near-term available quantum computing hardware.
As near-term quantum devices approach sizes of more than 50 qubits, a large amount of research has been devoted to finding tasks where such devices can outperform classical computers.
One area of research concerns so-called ``quantum supremacy'' \cite{preskill2012quantum, aaronson2016complexity, harrow2017quantum, neill2018blueprint}, which is about exhibiting a task for which the classical simulation is conjectured to be hard but which is performed efficiently on a quantum device. While the theoretical guarantees are sound, usually such tasks do not have straightforward practical applications, such as in the case of Boson sampling \cite{aaronson2011computational} and IQP circuits \cite{bremner2010classical}. On the other hand,  many investigations focus on finding applications for near-term quantum computers. Such applications are believed to be in quantum chemistry, optimization and machine learning, and possible algorithmic candidates are the variational quantum eigensolver (VQE) \cite{peruzzo2014variational, mcclean2016theory, kandala2017hardware} and quantum approximate optimization (QAOA) \cite{farhi2014quantum, farhi2016quantum}.
A good definition of near-term quantum computing for applications is to find quantum algorithms that minimize the number of qubits, the number of gates, the complexity of the quantum gates (for example, in terms of controlled operations), and are tailored to the available hardware, while at the same time the problem should be of significant practical relevance.

Linear systems are important in a large variety of applications in engineering and sciences. Generically, a linear system is specified by a non-square matrix $A \in \mathbbm R^{M\times N}$ and  a right-hand side vector $ b \in \mathbbm R^{M}$.  The task is to find a solution vector $ x \in \mathbbm R^N$ for which $A x=  b$.
Depending on the dimensions $M$ and $N$ and the rank of the matrix, the task of solving the linear system takes on various forms.
First, if the matrix is square and invertible, we can use matrix inversion to solve the linear system to find a unique solution.
If the matrix is square and non-invertible, the pseudoinverse inverts only non-zero eigenvalues.
If the matrix is non-square, we have the case of overdetermined and underdetermined equation systems.
An overdetermined equation system appears for example in regression, where a few parameters, given by the vector $x$, are used to explain a larger amount of data points, specified by the vector $b$.
In this case, no exact solution is possible and one often minimizes the $\ell_2$ norm $\norm{A  x -  b} _2^2$ to find the best solution.
On the other hand, if the equation system is underdetermined, an infinite set of solutions exists.  Further constraints can be imposed to find specific solutions, such as minimizing the $\ell_2$-norm of the solution via the pseudoinverse, or achieving sparsity of the solution via additional $\ell_0$ or $\ell_1$ norm constraints, such as in the LASSO estimator \cite{tibshirani1996regression} or compressed sensing \cite{candes2004robust}.

In this work, we study near-term quantum algorithms for solving linear systems.
We start by analyzing the use of basic variational algorithms for this task.
In variational algorithms, the quantum computer is used to prepare candidates for the solution vector, using a shallow sequence of parameterized quantum gates. Then, measurements are performed on the solution candidate to evaluate the quality of the candidate defined in terms of a loss function.
Finally, an optimization loop updates the variational parameters to improve the quality of the solution candidate.
We propose two types of Ans\"atze for these basic variational algorithms.
The first one uses Ans\"atze that are hardware-efficient, and without explicit usage of the matrix $A$ and ${b}$, hence is called the ``Agnostic" Ansatz.
Such Agnostic Ans\"atze can be used in various forms of optimization methods, such as Nelder-Mead method \cite{lagarias1998convergence}, imaginary-time evolution \cite{li2017efficient, mcardle2019variational}, or adiabatic-inspired optimization \cite{farhi2000quantum}.
The second Ansatz, an Alternating Operator Ansatz, is strongly dependent on the linear system and fully uses $A$ and $b$ via Hamiltonian simulation \cite{Lloyd1996, berry2015simulating, berry2015hamiltonian, low2017optimal, low2019hamiltonian}. This approach is inspired by the adiabatic approach and the QAOA method. The drawback is the more far-term nature of the approach.
Very recently, there have been other works that investigate the use of variational algorithms to solve linear systems \cite{Xiaosi2019variational, Carlos2019variational, Dong2019variational}, which share similarities with the above ideas.

We encounter simple types of linear systems which exhibit a potential problem for variational approaches.
For these linear systems, any pre-specified Ansatz
with a polynomial number of variational parameters will have a mostly flat optimization landscape with large loss function, which can be considered as a type of plateau effect.
A flat optimization landscape will thwart any optimization techniques to find the optimum.
We analyze various efforts that circumvent the known barren plateau issue \cite{mcclean2018barren}, such as properly-initialized imaginary time evolution or the adiabatic-inspired optimization.
We show that these attempts may still show the plateau effect,
which motivates further efforts and alternative ideas to solve large-scale linear systems with near-term quantum devices and achieve quantum advantage.

To provide a potential solution to the aforementioned problems, we pursue a different route and propose a different class of algorithms for solving linear systems on near-term quantum devices.
These algorithms are based on a concept we call Classical Combination of Variational Quantum States (abbreviated as CQS). As the name suggests, we show that different quantum states, variational or not, can be combined via classical pre- and post-processing to increase the power of near-term quantum computers.
This combination avoids difficulties in optimizing the variational parameters and provides some provable guarantees for solving linear systems.
We introduce the notion of an Ansatz tree and show how to find a good set of quantum states and the optimal combination coefficients. We show that at enough depth the Ansatz tree is guaranteed to include the solution, and also study the case of Tikhonov regularized regression. We introduce a heuristic approach for judiciously pruning the tree in the less important branches and expanding the tree in the more important branches. Our proposed heuristics is called the gradient expansion heuristics, and can be considered as just one example of many potential heuristics for CQS and the Ansatz tree.
To demonstrate the potential of this class of algorithms, we have conducted numerical simulations for solving linear systems with sizes up to $2^{300} \times 2^{300}$.
Finally, we show that a variation of the CQS method can achieve a similar provable guarantee as existing quantum algorithms for linear systems and can improve upon a recent work \cite{subacsi2019quantum}, reducing the quantum gate count by $(1 / \epsilon)$-fold, where $\epsilon$ is the desired error to the optimal solution (e.g., $\epsilon = 0.01$), while maintaining the use of only one ancilla qubit.

\section{Classical and Quantum Setting}

We are given a Hermitian matrix $A\in \mathbbm C^{N\times N}$ with spectral radius $\rho(A) \leq 1$. Assume without loss of generality that $N=2^n$. For non-Hermitian matrices, the standard Hermitian embedding can be used.
The right-hand side vector $ b \in \mathbbm C^{N}$,
which we assume to be normalized and write in quantum notation as $\ket b$. The normalization is not a major restriction as we could rescale the length of $x$.
The main task is to find a vector $x$  that solves the system of equations
\begin{equation} \label{linear}
A  x =  \ket b.
\end{equation}
In the quantum setting, we have to make assumptions about the access to the linear system.
First, we require quantum access to the right-hand side vector, i.e., a quantum circuit that prepares the vector $\ket b$ as a quantum state.
\begin{assume} \label{assumeB}
Assume availability of an efficient $n$-qubit quantum circuit described by the unitary $U_b$ such that
$U_b \ket{0^n} = \ket b$.
\end{assume}
Next, we require access to the matrix defining the linear system. Here, our main assumption is that the matrix is given by a small linear combination of known unitaries. This assumption is weaker than the assumption of an efficient Pauli decomposition.
\begin{assume}\label{assumeAUnitary}
Assume an efficient unitary decomposition of the matrix $A\in \mathbbm C^{N\times N}$, i.e.
$A = \sum_{k=1}^{K_A} \beta_k U_k$,
with $K_A = \Ord{{\rm poly}(\log N)}$ and unitaries $U_k \in \mathbbm C^{N\times N}$ with known efficient quantum circuits. We can always absorb the phase of $\beta_k$ into $U_k$, so we can assume $\beta_k > 0$.
\end{assume}

Next, we discuss the main loss functions used in this work. The first loss function is the well-known $\ell_2$-norm loss used in regression methods. This loss function is convex in $x$.
\begin{defn}
\label{defLossReg}
Let the linear system be given by $A\in \mathbbm C^{N \times M}$ and $ \ket b\in \mathbbm C^{N}$.
Define the loss function
$L_{R}( x) := \norm{  A  x -   \ket b }_2^2 =  x^{\dagger} A^\dagger A  x - 2 \realpar{ \bra b  A  x }  + 1.
$
\end{defn}
Instead of Definition \ref{defLossReg}, one may want to use the regularized version.
This version is common in statistics and machine learning, and is known as Tikhonov regularization \cite{ng2004feature}, or ridge regression \cite{hoerl1970ridge}.
\begin{defn}
\label{defLossTik}
Let the linear system be given by $A\in \mathbbm C^{N \times M}$ and $ \ket b\in \mathbbm C^{N}$.
Define the loss function
$L_T(x) := \frac{1}{2}\norm{x}_2^2 + \norm{Ax- \ket b}_2^2$.
\end{defn}
The regularization turns the loss function into a strongly convex function. This loss function will be used to prove a faster convergence of the Ansatz tree approach in Proposition \ref{prop:Ansatzproof}. The third loss function is obtained by defining a Hamiltonian which has a unique ground state that is the solution to the linear system.
This definition borrows techniques presented in \cite{subacsi2019quantum} for solving the linear system via a method inspired by adiabatic quantum computation.
To keep the adiabatic Hamiltonian positive across the adiabatic sweep, an ancilla has been introduced in Definition~\ref{defHamiltonianA}.
\begin{defn}\label{defHamiltonianA}
Let $A\in \mathbbm R^{N\times N}$ be symmetric and invertible with $\rho(A) \leq 1$.
Define an extended matrix
$A(s) := (1-s) Z \otimes \mathbbm 1 + s X \otimes A$.
In addition, define the parameterized Hamiltonian
$H(s) := A(s) P_{+,b}^\perp A(s)$,
with the projector $P_{+,b}^\perp := \mathbbm 1 - \ket{+,b} \bra{+,b}$.
\end{defn}
Among other properties, in \cite{subacsi2019quantum} it was shown that $H(1)$ has a unique ground state with zero eigenvalue given by
$\ket{+}\ket {x^{\ast}} =\ket{+} \frac{ A^{-1} \ket{ b}}{\norm{ A^{-1} \ket {b} }_2}$, which is proportional to the solution $A^{-1} \ket {b}$ after removing the ancilla.
This Hamiltonian implies the following loss function.
\begin{defn}\label{defLoss}
Define the loss function
$L_{H}(\ket x) := \bra {+,x} H(1) \ket {+,x}$.
\end{defn}
The loss function can also be written as
$L_{H}(\ket x) = \bra x  A^2 \ket x - \bra x A \ket{ b} \bra{ b}  A \ket x$ without the ancilla.

\section{Variational algorithms and ans\"atze}

We first discuss basic variational algorithms for solving linear systems.
A typical variational algorithm works as follows: one prepares multiple copies of a parameterized quantum state Ansatz and measures observables on it; the measurement results provide an estimate of the loss function. An optimization loop changes the parameters of the Ansatz with the goal of minimizing the loss function.
In this section, we consider two types of variational Ans\"atze. Different Ans\"atze require different assumptions on the available hardware and can lead to different sets of solutions.
\begin{itemize}
\item Agnostic Ansatz: We take Ans\"atze which perform single qubit rotations and entangling operations. We do not take into account information of the linear system itself except by measuring the loss function.
\item Alternating Operator Ansatz: We alternate the use of operators constructed from $A$ and the vector $\ket b$ for generating the Ansatz.
This requires Hamiltonian simulation of operators derived from $A$ and $\ket{b}\bra{b}$.
\end{itemize}
In particular, we focus on minimizing the Hamiltonian loss function $L_H(\ket{x})$, which is equivalent to finding the ground state of the Hamiltonian $H(1)$.
This allows the use of tools such as variational quantum eigensolver in quantum chemistry to solve linear systems of equations.
The detailed procedure to measure the Hamiltonian loss function is discussed in Appendix~\ref{secMeasure}.

\subsection{Details on variational algorithms for optimizing the Ansatz}
\label{secVariationalDetail}

We first discuss the details of variational algorithms.
We consider an Ansatz generated by a quantum circuit parametrized by $\theta$, i.e., $\ket{\psi(\theta)} = U_{\rm Ansatz}(\theta) \ket {0^n}$.
First, we show the basic variational quantum eigensolver (VQE) for finding the ground state of a Hamiltonian $H$.
Initialize the variational parameters $\theta$ to be $\theta_{\rm init}$.
While $\theta$ has not converged, do the following steps:
\begin{enumerate}
\item Prepare quantum state $\ket{\psi(\theta)}$ on the quantum computer.
\item Obtain an estimate for the loss function defined by $\bra{\psi(\theta)} H \ket{\psi(\theta)}$.
\item Update $\theta$ according to the obtained estimate of the loss function (e.g., using Nelder-Mead).
\end{enumerate}

In addition to using Nelder-Mead, another strategy for optimizing the variational parameters is through the use of imaginary time propagation.
Ideally, imaginary time propagation will move a given initial state to the ground state of the Hamiltonian as all excited states will be quickly suppressed.
As Ref.~\cite{McArdle2018} shows, instead of propagating the quantum system, one can directly propagate the parameters $\theta$. The detailed algorithm works as follows.
Set $\theta(0) = \theta_{\rm init}$. For $t=0, \delta t, 2\delta t, \cdots,T$, do the steps:
\begin{enumerate}
\item Obtain an estimate for all terms $C_i(t)$ and $M_{ij}(t)$ using copies of $\ket{\psi(\theta(t))}$, where $M_{ij}(t) = \realpar{ \left(\frac{\partial}{\partial \theta_i} {\psi(\theta(t))} \right )^\dagger \frac{\partial}{\partial \theta_j} \ket {\psi(\theta(t))} }$ and
$C_{i}(t) = \realpar{ \left(\frac{\partial}{\partial \theta_i} \ket {\psi(\theta(t))} \right )^\dagger H(s) \ket{\psi(\theta(t))} }$.
\item Perform variational imaginary time propagation: $\theta (t + \delta t) \leftarrow \theta (t) - M^{-1}(t) C(t) \delta t$.
\end{enumerate}

An approach to improve convergence to the solution is based on adiabatic evolution. Adiabatic evolution gradually changes the Hamiltonian $H(s)$ from the initial Hamiltonian (at $s=0$) to the target Hamiltonian (at $s=1$).
In the VQE setting, an adiabatic-assisted optimization was used in \cite{wecker2015progress, garcia2018addressing}. We follow \cite{garcia2018addressing} and refer to this approach as the adiabatic-assisted VQE (AAVQE).
To implement AAVQE, we discretize $s$ into $T$ adiabatic steps, $s_0 = 0, s_1, \ldots, s_{T-1}, s_T = 1$.
At each adiabatic step $t$, we use the optimized variational parameter $\theta^*_{t-1}$ for $H(s_{t-1})$ as the initial guess for $H(s_{t})$.
We first initialize $\theta$ to be $\theta^*_0$, where $\ket{\psi(\theta_0^*)}$ is the ground state for the initial Hamiltonian $H(0)$.
Then, for $t=1, \cdots, T$, we perform Nelder-Mead or imaginary time propagation on the parameter $\theta$ to find an optimized variational parameter $\theta^*_{t}$ for $H(s_t)$ by starting from $\theta^*_{t-1}$.

\subsection{Agnostic Ansatz}
We consider a pre-specified Ansatz with several layers, where each layer consists of single-qubit rotation for every qubit and a set of controlled NOT (CNOT) gates for entanging different qubits.
The variational parameters are the rotation angles in the single-qubit rotations.
This Ansatz does not take explicit account of the linear systems $A, b$ and hence we use the name Agnostic Ansatz.
We have performed numerical experiments on the Rigetti quantum virtual machine \cite{Rigetti_pyquil} with system sizes up to $N=16$ and explored various patterns of how the CNOT gates are applied.
We observed that most CNOT gate patterns are able to find the solution as one increases the number of layers, and hence also increases the number of variational parameters.
We have also tested the adiabatic-assisted VQE algorithm \cite{garcia2018addressing,wecker2015progress}. The average accuracy (average fidelity of the output vector with the actual solution over randomly generated linear systems) approaches unity as one increases the number of adiabatic steps. In particular, we observed an improvement using adiabatic-assisted VQE over standard VQE. For plots and detailed findings, please refer to Appendix \ref{appAgnostic}.

\subsection{Alternating Operator Ansatz}
\label{sec:directAnsatz}

We now discuss a different Ansatz that contains information about the linear system, i.e., the matrix $A$ and the vector $b$. This Ansatz comes at the cost of requiring Hamiltonian simulation of operators involving $A$ and $b$. The Ansatz is inspired by the method presented in \cite{subacsi2019quantum} for solving the linear system via adiabatic techniques.
We can write out the Hamiltonian in Definition \ref{defHamiltonianA} as
\be
H(s) &=& (1-s)^2 \mathbbm 1 + s^2 \mathbbm 1 \otimes A^2 - (1-s)^2 \ket{-,b}\bra{-,b} - s^2 (\mathbbm 1 \otimes A)\ket{+,b} \bra{+,b}(\mathbbm 1 \otimes A) \nonumber\\&& - s(1-s)\left(\ket{+,b}\bra{-,b} (\mathbbm 1 \otimes A) + (\mathbbm 1 \otimes A)\ket{-,b}\bra{+,b} \right).\nonumber
\ee
Examining the Hamiltonian leads to four Hermitian operators that make up $H(s)$, which are scaled by combinations of $s$ and $\pm(1-s)$. The Hamiltonians are $H_1=A^2$, $H_2=\ket{-,b}\bra{-,b}$, $H_3=(\mathbbm 1 \otimes A)\ket{+,b} \bra{+,b}(\mathbbm 1 \otimes A)$, and $H_4=\ket{+,b}\bra{-,b} (\mathbbm 1 \otimes A) + (\mathbbm 1 \otimes A)\ket{-,b}\bra{+,b}$, aside from the identity matrix which only shifts the spectrum and induces a global phase in the dynamics.

Based on the four Hamiltonians $H_1, H_2, H_3,$ and  $H_4$, we can construct an Alternating Operator Ansatz, which is a direct translation of the approach in \cite{subacsi2019quantum} into the QAOA framework \cite{farhi2014quantum}.
Let us define the Ansatz as follows.
Let $p$ be the number of layers of alternating unitaries.
For a set of variational parameters $\theta_{k,j}$, $k\in [p]$ and $j\in [4]$,
we define the parameterized unitaries corresponding to the four Hamiltonians,
$U_j(\theta_{k,j}) := e^{-i \theta_{k,j} H_j}$.
Then our variational Ansatz is
\be
U_4(\theta_{p,4}) U_3(\theta_{p,3}) U_2(\theta_{p,2}) U_1(\theta_{p,1}) \dots U_4(\theta_{1,4}) U_3(\theta_{1,3}) U_2(\theta_{1,2}) U_1(\theta_{1,1}) \ket{b}.\nonumber
\ee
This Ansatz contains explicit information of $A$ and $b$, which avoids potential problems in variational algorithms discussed in Section~\ref{sec:potprob}.
However, the suitability of this method for the use in near-term quantum computers depends on the difficulty of simulating the unitaries $U_j(\theta_{k,j})$.
Given Assumptions \ref{assumeB} and \ref{assumeAUnitary}, we can express the Hamiltonians as
$H_1=\sum_{k,k'=1}^{K_A} \alpha_k \alpha_{k'} U_{k'} U_k$, $H_2=(\mathbbm 1\otimes U_b) \ket{-, 0^n}\bra{-, 0^n}(\mathbbm 1\otimes U_b^\dagger)$,
$H_3=\sum_{k,k'=1}^{K_A} \alpha_k \alpha_{k'} U_k (\mathbbm 1\otimes U_b) \ket{+, 0^n}\bra{+, 0^n}(\mathbbm 1 \otimes U_b^\dagger) U_{k'}$, and
$H_4=\sum_{k=1}^{K_A} \alpha_k (\mathbbm 1\otimes U_b) \ket{+, 0^n}\bra{-, 0^n}(\mathbbm 1\otimes U_b^\dagger) (\mathbbm 1 \otimes U_k) + (\mathbbm 1 \otimes U_k) (\mathbbm 1\otimes U_b) \ket{-,0^n}\bra{+,0^n} (\mathbbm 1\otimes U_b^\dagger) $. As these operators are combinations of unitaries and projectors,  Hamiltonian simulation for simple cases may be within the realm of near-term hardware.
However, at this point, there are no guarantees on performance due to the potentially difficult optimization of the variational parameters $\theta_{k, j}$.

\subsection{Potential problems in variational algorithms for solving linear systems}
\label{sec:potprob}

Typical optimization for linear systems minimizing $\norm{Ax - \ket b}^2_2$ from Definition \ref{defLossReg} is convex in $x$, and hence is easy to solve in principle.
This is because the gradient is larger when we are further away from the optimal solution and the negative gradient always points in the descent direction.
On the other hand, when we restrict to the variational quantum state space, the optimization landscape is no longer convex and is poorly understood.
We consider toy classes of linear systems that show difficulties when using the variational methods.
The difficulties arise essentially from the fact that a random linear system in an exponentially large Hilbert space will have solutions that only have exponentially small overlap with the Ansatz $\ket{\psi(\theta)}$ under most parameters $\theta$. While this may extend to a broader class of linear systems than discussed here, we note that this argument does not preclude the existence of certain classes of linear systems and Ans\"atze in combination with properly chosen loss functions for which variational optimization can provide quantum advantages.  Linear systems with additional structure may be one example and a better characterization of such linear systems is left for future work. In section \ref{sec:CQS}, we provide a method to overcome the issues discussed here.

Consider the following toy problem. Let $k \in \{0,1\}^n$ be an arbitrary $n$-bit string, and $U_b$ be the quantum circuit for generating the state $\ket{b}$. Let the problem be given by
\be \label{eq:varcounterex}
A &=& (\sigma_x^{(1)})^{k_1}\otimes \cdots \otimes (\sigma_x^{(n)})^{k_n},\quad \ket b = U_b \ket {0^n}.
\ee
When $U_b$ is $\mathbbm 1$, the solution to the equation $A \ket x = \ket b$ is simply $\ket {k}$. Note that $A$ is sparse and the condition number of $A$ is $1$, hence existing quantum algorithms for linear systems \cite{Harrow2009, ambainis2010variable, childs2017quantum} are able to solve this linear system efficiently.
Now assume a variational Ansatz $\ket{x(\theta)}$ that contains the solution (up to global phase) for every $k \in \{0,1\}^n$.
For example, one possible, but not the only, choice would be
\be
 \ket{x(\theta)} = e^{-i\theta_1 \sigma_x^{(1)}} \otimes \cdots \otimes  e^{-i\theta_n \sigma_x^{(n)}} \ket {0^n},
 \label{eq:varanscounterex}
\ee
with $m = n$ variational parameters $\theta$.
We now show that for both loss functions Definition \ref{defLossReg} (the regression loss function) and Definition \ref{defLoss} (the Hamiltonian loss function), no matter what the initial $\theta$ is, the loss function will be flat at that point with an exponentially small slope with high probability.

\begin{figure}[t]
    \centering
    \includegraphics[width=0.98\textwidth]{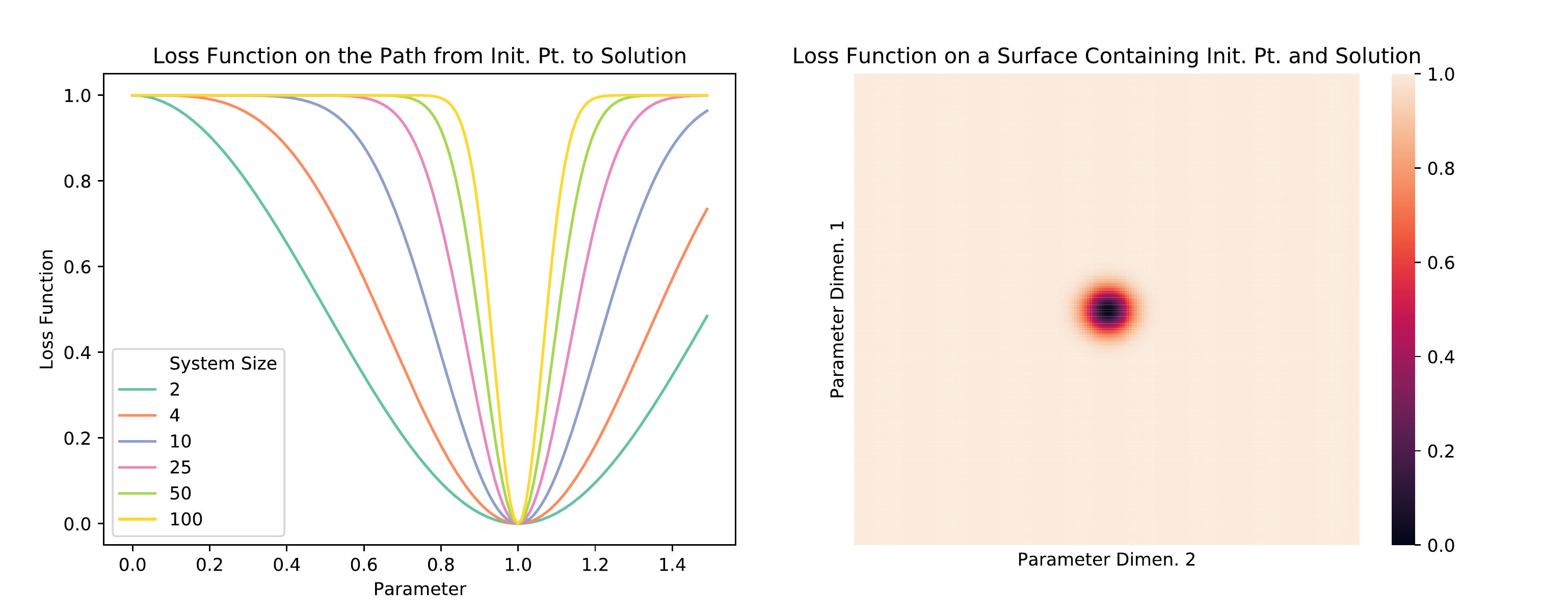}
    \caption{The optimization landscape of variational linear system solver for $k \in \{0, 1\}^n$ with $\norm{k}_0 = \lceil n / 2 \rceil$ and
    $A = (\sigma_x^{(1)})^{k_1}\otimes \cdots \otimes (\sigma_x^{(n)})^{k_n}$, $\ket{b} = \ket{0^n}$ using the variational Ansatz Eq.~\eqref{eq:varanscounterex}. Left: Here, we plot a one-dimensional cut through the high-dimensional landscape, tracing along a line connecting the initial point (parameter $= 0$) and the optimal solution (parameter $= 1$) for varying system size. We can clearly see the appearance of a plateau with the solution being a sharp valley at some point in the high-dimensional surface. Right: We plot the landscape on a surface that contains the initial point (bottom left corner with loss function $= 1.0$) and the optimal solution (the middle point with loss function $= 0.0$) for a system with $n=100$. The loss function is near flat everywhere except a sharp hole in the middle that contains the solution.}
    \label{fig:losslandscape}
\end{figure}

For $ \norm{A\ket{x(\theta)} - \ket{b}}_2^2$ from Definition \ref{defLossReg}, consider an initial point $\theta^0$, the local expansion is
$$2 - 2 \realpar{\bra{b} A \ket{x(\theta^0)}} - \sum_{i=1}^m 2 \realpar{\bra{b} A \frac{\partial}{\partial \theta_i} \ket{x(\theta^0)} } \delta \theta_i + \Ord{m^2 \delta \theta_i^2}.$$
The loss function at the initial point $\theta^0$ is $2 - 2 \realpar{\bra{b} A \ket{x(\theta^0)}}$ and the gradient in the $i$-th parameter $\theta_i$ is $-2 \realpar{\bra{b} A \frac{\partial}{\partial \theta_i} \ket{x(\theta^0)} }.$
We usually consider an Ansatz $\ket{x(\theta)}$ that changes slightly when $\theta$ differs by a small amount, e.g., when $\theta$ are the rotation angles of single-qubit rotations. Formally, this means $\norm{\frac{\partial}{\partial \theta_i} \ket{x(\theta^0)}} \leq G$, where $G$ is some constant.
Now, we will show that the gradient in each direction is exponentially small.
First, because the norm of $\frac{\partial}{\partial \theta_i} \ket{x(\theta^0)}$ is bounded by $G$, we can see that there are at most $2^{n/2}$ entries in $\frac{\partial}{\partial \theta_i} \ket{x(\theta^0)}$ with squared absolute value $ \geq G^2 / 2^{n / 2}$.
So the probability of $|\bra{b} A \frac{\partial}{\partial \theta_i} \ket{x(\theta^0)}|^2 \geq G^2 / 2^{n / 2}$ is at most $1 / 2^{n/2}$.
By union bound, with probability at least $1 - m / 2^{n/2} \approx 1$, we will have
$$\left|\bra{b} A \frac{\partial}{\partial \theta_i} \ket{x(\theta^0)}\right|^2 < G^2 / 2^{n / 2}, \,\mbox{ for all }\, i \in \{1, \ldots, m\}.$$
Hence the gradient in the $i$-th variational parameter $\bra{b} A \frac{\partial}{\partial \theta_i} \ket{x(\theta^0)}$ will be exponentially small for all $i = 1, \ldots, m$ with high probability as long as $m \ll 2^{n/2}$.
Hence, around $\theta^0$, the loss function $\norm{A\ket{x(\theta)} - \ket{b}}_2^2$ is flat and has an almost constant value $2 - 2 \realpar{\bra{b} A \ket{x(\theta^0)}}$.

The analysis is analogous for $\bra{x(\theta)} (A^2 - A\ket{b}\bra{b}A) \ket{x(\theta)}$ from Definition \ref{defLoss}.
Pick an initial point $\theta^0$, the local expansion of the loss function at $\theta^0$ is
$$1 - \bra{x(\theta^0)} A\ket{b}\bra{b}A \ket{x(\theta^0)} - \sum_{i=1}^m 2 \realpar{\bra{x(\theta^0)} A\ket{b}\bra{b}A \frac{\partial}{\partial \theta_i} \ket{x(\theta^0)}} \delta \theta_i + \Ord{m^2 \delta \theta_i^2}.$$
We have that $\bra{b}A \frac{\partial}{\partial \theta_i} \ket{x(\theta^0)}$ is exponentially small for all $i$ with high probability.
Thus the loss function is flat with a function value of $1 - \bra{x(\theta^0)} A\ket{b}\bra{b}A \ket{x(\theta^0)}$ around $\theta^0$.
The flat region is not a local minimum, but a plateau with large loss function.

For the example linear system defined in Eq.~(\ref{eq:varcounterex}), this plateau problem holds no matter how the variational circuit for generating $\ket{x(\theta)}$ is structured.
The behavior appears even for shallow circuits, as in Eq.~(\ref{eq:varanscounterex}).
The behavior only becomes evident when the system size is large compared to the number of variational parameters, i.e., $m \ll 2^{n/2}$, but still within reach in the NISQ era \cite{preskill2018quantum} with few tens of qubits.
A numerical experiment that demonstrates this behavior can be seen in Figure~\ref{fig:losslandscape}. From the figure, we can clearly see the appearance of the plateau as the system size grows larger.

An intuitive view on the failure in this toy example is that a rank-$1$ projection, such as $\bra{x(\theta)} A \ket{b} \bra{b} A \ket{x(\theta)}$ in the Hamiltonian loss function, would be exponentially small, so it gives little signal on how to find the solution.
In \cite{larose2019variational, khatri2019quantum, bravo2019variational}, another type of loss function, called the local loss function, have been defined to ameliorate this problem.
In the context of linear systems, it is defined as $\bra{x(\theta)} A U_b (\mathbbm 1 - \frac{1}{n} \sum_i \ket{0_i}\bra{0_i}) U_b^\dagger A \ket{x(\theta)}$, where $U_b$ is the circuit to generate the state $\ket{b}$ and $\ket{0_i}\bra{0_i}$ is the projection onto the $i$-th qubit.
Intuitively, this prevents the projection from being exponentially small because it is now a sum of single qubit projections.
When $U_b = \mathbbm 1$, and $A = (\sigma_x^{(1)})^{k_1}\otimes \cdots \otimes (\sigma_x^{(n)})^{k_n}$, this local loss function could indeed provide sufficient signals to solve the problem.
However, a similar plateau problem may appear in cases where $A$ and $U_b$ contains many entangling circuits.
In Appendix~\ref{app:moreprob}, we will show that when $A$ and $U_b$ are random polynomial-sized quantum circuits, the local loss function plateaus at the value $1/2$.
The basic idea is that the single-qubit reduced density matrix of an entangled quantum state will be very close to a completely mixed state. Even though single qubit projection $\ket{0_i} \bra{0_i}$ will not be exponentially small, it will be sharply concentrated at the value $1/2$ with an exponentially small variance.
While the local loss function would have a flat landscape when $A$ and $\ket b$ are based on polynomial-sized quantum circuits sampled randomly, if $A$ and $\ket b$ have more structure the local loss function can provide benefits for the variation optimization \cite{larose2019variational, khatri2019quantum, bravo2019variational}.

In all cases when the loss function landscape is essentially flat,  the presence of small errors due to statistical fluctuations in the quantum measurements would make it very hard for existing optimization approaches to find the optimal solution efficiently even if there exists a solution in the Ansatz.
For example, if we use variational imaginary time evolution as described in Section \ref{secVariationalDetail} to optimize the variational parameters, the same analysis shows that $C(t)$ used in the propagation of $\theta$ would be an exponentially small vector.
This means even if $C(t)$ is measured to exponential-precision (which already requires exponential time due to the statistical error in quantum measurements), the imaginary time propagation would still take exponential time to find the ground state.

\begin{figure}[t]
    \centering
    \includegraphics[width=0.98\textwidth]{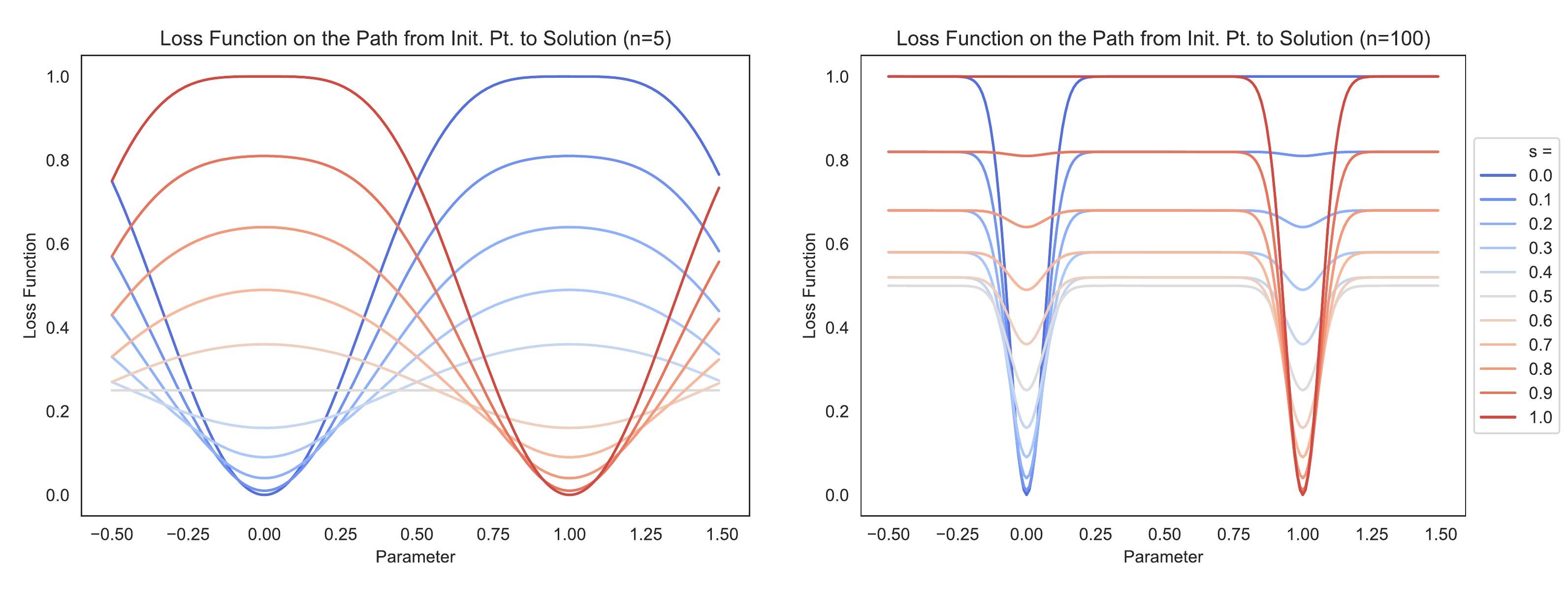}
    \caption{The optimization landscape under adiabatic evolution ($s$ from $0$ to $1$) to solve linear systems for $k \in \{0, 1\}^n$ with $\norm{k}_0 = \lceil n / 2 \rceil$ and
    $A = (\sigma_x^{(1)})^{k_1}\otimes \cdots \otimes (\sigma_x^{(n)})^{k_n}$, $\ket{b} = \ket{0^n}$ using the variational Ansatz Eq.~\eqref{eq:varanscounterex}. Here, we plot a one-dimensional slice through the high-dimensional landscape on a line passing through the initial point (parameter $= 0$) to the optimal solution (parameter $= 1$). Left: A system size of $5$ qubits. Right: A system size of $100$ qubits. For small system size, after $s > 0.5$, the variational parameter will be able to move toward the solution. For large system size, the adiabatic evolution will continue to be stuck at the initial point and end up at a plateau when $s = 1.0$. }
    \label{fig:losslandscapeadia}
\end{figure}

One important note is that this plateau effect originates from a different cause compared to the barren plateau problem \cite{mcclean2018barren}, which appears in Ansatz initialized as a random quantum circuit with enough depth.
Hence previous attempts to evade barren plateaus may still face this problem.
For example, one may perform adiabatic evolution to avoid barren plateaus, as suggested from Definition \ref{defHamiltonianA}. Start with $\ket{x(\theta^0)} = \ket{-, b}$, which is the ground state of $H(s=0) = A(0)(\mathbbm 1 - \ket{+, b}\bra{+, b})A(0)$, with $A(s) = (1-s)Z \otimes \mathbbm 1 + s X \otimes A$; then gradually change $s$ from $0$ to $1$ and use variational imaginary time evolution to maintain the quantum state in the ground state.
Intuitively, when $s$ is changed slightly, the ground state of $H(s)$ will only shift by a small amount, so performing imaginary time evolution allows to follow closely the adiabatic path.
This intuition is true in the original exponential-sized Hilbert space, but does not hold in the polynomial-sized variational parameter space.
The reason is that the topology in the polynomial-sized variational parameter space is very different from the original exponential-sized Hilbert space.
For example, let $a, b, c \in \{0, 1\}^n,$ and $\epsilon \ll 1$, then $\ket{a} + \epsilon \ket{b}$ is close to $\ket{a} + \epsilon \ket{c}$ in the original Hilbert space, but they may be very far away in the polynomial-sized variational parameter space.

Let us consider the same toy example as in Eq.~(\ref{eq:varcounterex}) and the loss function analogous to Definition \ref{defLoss} is
$$\bra{x(\theta)} H(s) \ket{x(\theta)} = c(s) - (1-s)^2 f(\theta) - s^2 g(\theta) - 2 s(1-s) h(\theta), $$
where $c(s) = 1-2s+2s^2, f(\theta) = |\braket{-, 0^n}{x(\theta)}|^2, g(\theta) = |\bra{+, b}(\mathbbm 1 \otimes A)\ket{x(\theta)}|^2, h(\theta) = \frac{1}{2} \braket{{x(\theta)}}{-, 0^n} \bra{+, b}(\mathbbm 1 \otimes A)\ket{x(\theta)} + \frac{1}{2}\bra{{x(\theta)}}(\mathbbm 1 \otimes A)\ket{+, k} \braket{-, 0^n}{x(\theta)}$.
At $s=0$, the optimum is $\ket{-, 0^n}$. At $s=1$, the optimum is $\ket{+, k}$. The adiabatic evolution starts at $s=0$ with $\ket{x(\theta^0)} = \ket{-, 0^n}$.
In the original Hilbert space, this loss function is quadratic in $\ket{x}$ and we can always move in a direction that decreases the loss function.
However, in the polynomial-sized variational parameter space ($m \ll 2^{n/2}$), the landscape of $g(\theta)$ and $h(\theta)$ will both be flat around $\theta^0$ due to the inner product $\bra{+, b}(\mathbbm 1 \otimes A)\ket{x(\theta)}$ and the previous analysis.
This means the loss function looks like
$c(s) - (1 - s)^2 f(\theta)$ locally around $\theta^0$
and is always minimized at $\ket{-, 0^n} = \ket{x(\theta^0)}$ for all $s \in [0, 1]$.
So even if $s$ has been changed to $\Delta > 0$, the variational parameter will still stay at $\theta^0$.
The adiabatic evolution in the variational parameter space would always stay at $\ket{x(\theta^0)} = \ket{-, 0^n}$ and fail to find the solution. An illustration of this analysis can be found in Figure~\ref{fig:losslandscapeadia}.

The conclusion is that for variational algorithms with a pre-specified Ansatz, it could be challenging to solve many linear systems due to the flat landscape.
This is not to say that variational quantum algorithms are hopeless in achieving quantum advantages for linear systems.
If we have a way of finding an Ansatz with an initial parameter that is close to the solution by making use of structure in $A, \ket{b}$, then this problem could be circumvented.
For example, an Ansatz that contains explicit application of $A$, such as $e^{-i A t}$, in its variational circuit does not fall into this problem.
The Alternating Operator Ansatz proposed in Section~\ref{sec:directAnsatz} is such a choice that may be able to avoid this problem, but require further study of convergence guarantees and the optimization of variational parameters may also be hard.
In the next section, we aim to propose an alternative approach that could circumvent the problem discussed here and offer some provable guarantees.

\section{Classical combination of quantum states  for linear systems}
\label{sec:CQS}
\subsection{Main idea}
Most hybrid quantum-classical algorithms parameterize the quantum state with classical parameters, and have the quantum state created on the quantum processor.
To extend the reach of near-term quantum devices, we consider an approach that broadens the set of manipulable states. Let the $N = 2^n$-dimensional Hilbert space be $\mathcal{H}$.
Let $ U$ denote an $n$-qubit circuit with $k$ real parameters, i.e., it may denote a circuit with alternating parameterized single-qubit rotations ($k$ in total) and all-to-all fixed CNOTs. Also let $\mathcal V_{U}= \{\ket{\psi_U({\theta})} \, | \, {\theta} \in \mathbb{R}^k\} \subset \mathcal{H}$ denote the corresponding quantum states.
A typical variational algorithm is based on a such a parameterized Ansatz $\mathcal V_{U}$ and tries to find the best $\ket{\psi_U({\theta})}$ in $\mathcal{V}_{U}$ by tuning ${\theta}$.
However, there are two drawbacks when using typical variational algorithms.
\begin{itemize}
    \item $\mathcal{V}_{U}$ may not be large enough to contain the solution. Such a drawback may often be the case in hardware-efficient Ans\"atze and even the Alternating Operator Ansatz discussed above.
    \item Even when $\mathcal{V}_{U}$ contains the solution, the variational parameter ${\theta}$ could be difficult to optimize. This problem can already be seen in the toy examples presented in previous sections.
\end{itemize}
Here, we improve on both drawbacks with a method called classical combination of variational quantum states (CQS).
In classical combination of variational quantum states, we consider a hybrid quantum-classical state.
Let $U_i$ for $i=1, \dots, m$ be quantum circuits with $k_i$ parameters each.
We construct a state vector ${x} \in \mathcal{H}$ as a quantum-classical hybrid,
$${x} = \sum_{i=1}^m \alpha_i \ket{\psi_{U_i}({\theta}_i)}, \mbox{ where }
\, \alpha_1, \ldots, \alpha_m \in \mathbb{C}, {\theta}_1  \in \mathbb{R}^{k_1}, \ldots, {\theta}_m \in \mathbb{R}^{k_m},$$
where $\alpha_i$ are the combination parameters and ${\theta}_i$ are the usual variational parameters.
Both parameters are stored on the classical processor.
However, the state vector ${x} \in \mathcal{H}$ is never created on the quantum processor.
Furthermore ${x}$ may not be normalized, so it is not a quantum state in general.

To manipulate ${x}$ with near-term quantum algorithms, the most important component is the ability to measure its expectation value for an observable $O$.
We can obtain the expectation value ${x}^\dagger O {x}$ by performing quantum measurements and classical post-processing, via the following steps.
\begin{enumerate}
    \item Estimate $\bra{\psi_{U_i}({\theta}_i)} O \ket{\psi_{U_j}({\theta}_j)}$ using a modified Hadamard test (see Proposition~\ref{propmodiHadamard} in Appendix~\ref{secMeasure}) on the quantum processor.
    The modified Hadamard test comes at the cost of preparing $\ket{\psi_{U_i}({\theta}_i)}$ and $\ket{\psi_{U_j}({\theta}_j)}$ in superposition using one additional ancilla.
    \item Compute $\sum_{i=1}^m \sum_{j=1}^m \alpha_i^* \alpha_j \bra{\psi_{U_i}({\theta}_i)} O \ket{\psi_{U_j}({\theta}_j)}$ on the classical processor.
\end{enumerate}
For comparison, suppose for the moment that ${x}$ is a (normalized) quantum state. To create ${x}$ on the quantum processor, we need at least $O(\log(m))$ ancilla qubits and $m$ controlled unitaries that prepare all $\ket{\psi_{U_i}({\theta}_i)}, \forall i$ in superposition.
Hence the improvement in terms of quantum resources using the hybrid quantum-classical state ${x}$ instead of the corresponding quantum state is
\begin{center}
    gate count: $m$ times $\rightarrow 2$ times, \\
    ancilla count: $\mathcal{O}(\log(m)) \rightarrow 1$.
\end{center}
The CQS method comes at the cost of many repetitions in quantum measurements.
However, we do not need to maintain quantum  coherence between measurements, hence it would be especially beneficial on noisy intermediate-scale quantum devices as the gate count is reduced  $m / 2$-fold.
For example, when we consider a classical combination of $300$ variational quantum states, then we can reduce the gate count by $150$ times.
On a near term quantum device, the gate count is often limited due to the error present in the device.
Hence the space of possible variational quantum states $\mathcal{V}_U$ that can be prepared without error on the quantum processor will be limited by the gate count.
The classical combination of quantum states thus provides a considerable improvement upon the space of manipulable states on near-term quantum processors.
An illustration of this idea is shown in Figure~\ref{fig:ccvqs}.

\begin{figure}[t]
    \centering
    \includegraphics[width=1.0\textwidth]{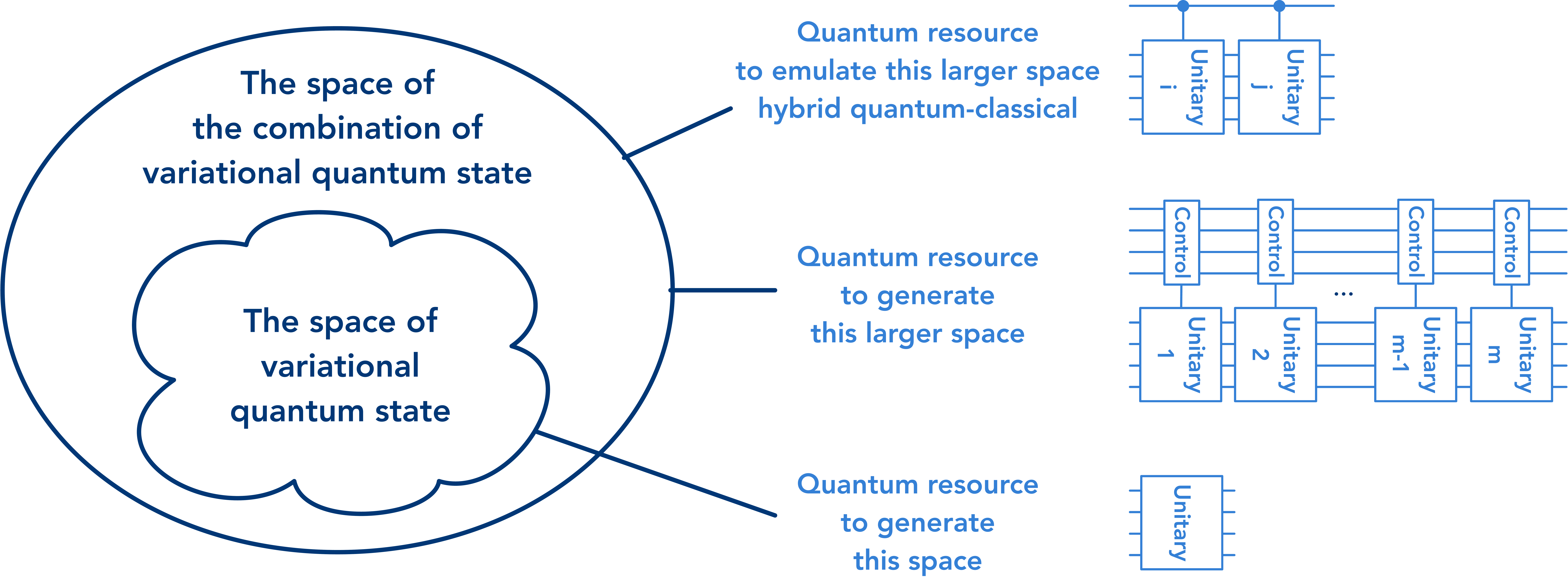}
    \caption{Illustration of our \textit{classical combination of variational quantum states} (CQS) approach. By considering subspaces spanned by $m$ variational quantum states, we are able to represent a larger class of states in the Hilbert space $\mathcal{H}$. The basic concept is illustrated on the left-hand side. In order to generate states in the $m$-dimensional subspace, we have to increase the quantum resources $m$-fold ($m$ times more gates and $O(\log(m))$ ancilla qubits that jointly control the unitaries). This case is illustrated by the middle picture on the right-hand side. Using a hybrid quantum-classical emulation, we can operate in this larger space using only a single additional ancilla qubit and twice as many gates, as illustrated in the top picture on the right-hand side.
    }
    \label{fig:ccvqs}
\end{figure}

We now present the meta strategy for finding an ${x} \in \mathcal{H}$ that solves the linear system of equations.
The approach consists of an optimization and an expansion step. The approach may avoid the optimization of ${\theta}_i$ which can involve a complicated optimization landscape.
We start with $m=1$ and a quantum state $\ket{\psi_{U_1}({\theta}_1)}$. Each iteration proceeds as follows.
\begin{enumerate}
    \item Optimization: Solve for the optimal $\alpha^*_1, \ldots, \alpha^*_m \in \mathbb{C}$ with ${x}(\alpha) = \sum_{i=1}^m \alpha_i \ket{\psi_{U_i}({\theta}_i)}$.
    \item Expansion: Using the current optimum ${x}(\alpha^*) = \sum_{i=1}^m \alpha^*_i \ket{\psi({\theta}_i)}$, find a next circuit $U_{m+1}$ with $k_{m+1}$ parameters and a setting $\theta_{m+1} \in \mathbbm R^{k_{m+1}}$ for those parameters. This circuit generates the state $\ket{\psi_{U_{m+1}}({\theta}_{m+1})}$.
    \item Set $m \leftarrow m+1$.
\end{enumerate}
A few comments are in order. The optimization Step 1 is convex and is described in the next Section \ref{sec:optcombparam}. The expansion Step 2 assumes that we can find a new circuit $U_{m+1}$. This circuit may or may not be parameterized by a parameter vector $\theta_{m+1} \in \mathbbm R^{k_{m+1}}$, the meta strategy can operate in both cases. In fact, the remainder of this section including the Ansatz tree approach in Section \ref{sec:ansatztree} does not explicitly discuss these parameters, but all states can also be thought of as parameterized.
In the case that the circuits are indeed parameterized, setting the parameters to useful values may involve optimization, which again may run into the plateau issues discussed above. However, the strategy can also be used without optimizing $\theta_{m+1}$ as long as the new state $\ket{\psi_{U_{m+1}}(\theta_{m+1})}$ is sufficiently different from the previous states.
Also note that the number of parameters $k_{m+1}$ may not have any particular relationship with the number of parameters of the previous steps $k_{1},\dots,k_{m}$.

\subsection{Optimization of combination parameters}
\label{sec:optcombparam}
We first focus on the case when we have selected a good set of $\ket{\psi_{U_1} }, \ldots, \ket{\psi_{U_m} }$, e.g., $A^{-1} \ket{b} \in \mbox{span}\{\ket{\psi_{U_1}}, \ldots, \ket{\psi_{U_m} }\}$, and we want to optimize over $\alpha_1, \ldots, \alpha_m$.
We will show that the optimization of $\alpha_1, \ldots, \alpha_m$ will always find the optimal solution.
This optimization is in stark contrast to the optimization over ${\theta}_i$, which can result in plateaus and local minima.
To simplify notation, we let $\ket{u_i} = \ket{\psi_{U_i}}$ further on.
In order to solve linear systems of equations, the standard regression loss function is
$$L_R(x):=\norm{A{x} - \ket{b}}_2^2 = {x}^\dagger A^\dagger A {x} - 2 \realpar{\bra{b} A{x}} + 1,$$
as in Definition \ref{defLossReg}.
Given ${x} = \sum_{i=1}^m \alpha_i \ket{u_i}$, we can reduce the optimization in an exponentially large space ${x} \in \mathcal{H}$ to an optimization over $m$ variables.
Let
$V = \left(v_1, \cdots, v_m\right)$
with the column vectors
$ v_i = A \ket {u_i}$.
We can now simply express the left-hand side of the linear system as
$
A x  = \sum_{i=1}^m \alpha_j A \ket {u_i} = V \alpha$.
Thus, we would like to minimize
\be
\norm{ V \alpha - \ket{b} }^2_2 = \alpha^\dagger V^\dagger V \alpha - 2 \realpar{ q^\dagger \alpha  }+ 1, \nonumber
\ee
where we introduced $q_i = \bra i  V^\dagger \ket b = \bra {u_i} A^\dagger \ket b$.
We obtain a simple regression problem for the combination parameters $\alpha$ with the kernel matrix
$(V^\dagger V)_{ij} = \bra {u_i} A^\dagger A \ket {u_j}$.
We can cast this quadratic optimization problem with complex variable $\alpha \in \mathbb{C}^m$ to a real optimization problem
$\min_{z} z^T Q z - 2 r^T z + 1$,
by letting $z = [\realpar{\alpha}, \impar{\alpha}] \in \mathbb{R}^{2m}$ and let
$$Q = \begin{pmatrix}
\realpar{V^\dagger V} & \impar{V^\dagger V} \\
\impar{V^\dagger V} & \realpar{V^\dagger V}
\end{pmatrix}, \,\, r = [\realpar{q}, \impar{q}].$$
Once all the input quantities $Q$ and $r$ are determined, such a regression problem can be solved with standard methods for convex quadratic programming. The inputs $Q$ and $r$ can be measured on a quantum computer using the strategies in Appendix \ref{secMeasure} and \ref{app:optcomb}. However, such measurements result in erroneous estimates of the quantities. The error will translate into an error in the loss function and the proposed solution for the combination parameters $\alpha$.
Using standard results in random matrix theory, we are able to achieve a rigorous bound on the error of the obtained solution, see Proposition~\ref{prop:guaranteesubspaceinformal}.
See also Appendix~\ref{app:optcomb} for a detailed analysis and Proposition~\ref{prop:guaranteesubspace} for the complete statement.

\begin{prop}[informal]
\label{prop:guaranteesubspaceinformal}
We can find an $\hat{\alpha} \in \mathbb{C}^m$ such that it is $\epsilon$-close to optimal, $$ L_R\left (\sum_i \hat{\alpha}_i \ket {u_i}\right) \leq \min_{\alpha_1, \ldots, \alpha_m \in \mathbb{R}} L_R\left (\sum_i \alpha_i u_i \right) + \epsilon,$$
using $\Ord{K_A^2 m^3 / \epsilon}$ measurements on the quantum device, where $A = \sum_{k=1}^{K_A} \beta_k U_k$ given in Assumption~\ref{assumeAUnitary}.
\end{prop}

A natural question that arises is whether the problem of solving linear system will become much easier when we consider optimization over a small subspace ${\rm span}(\ket{u_1}, \ket{u_2}, \ldots, \ket{u_m})$.
We can show that finding a near-optimal combination parameters in a subspace is BQP-complete.
\begin{prop}[informal]
Finding the combination parameters of $\ket{u_1}, \ket{u_2}, \ldots, \ket{u_m}$ to minimize $L_R\left (\sum_{i=1}^m \alpha_i\ket{u_i} \right)$ is BQP-complete.
\end{prop}
See Proposition~\ref{prop:bqpcompsubspace} in Appendix~\ref{app:optcomb} for the complete statement and proof.

\subsection{Ansatz tree approach for finding the subspace}
\label{sec:ansatztree}
We have shown good theoretical properties in the case where the subspace is fixed, such as a guarantee for finding a near-optimal solution in the subspace and BQP-completeness.
However, the results so far rely on already knowing a subspace that approximately contains the solution $x$.
Here, we propose an approach that constructs the subspace by exploring the space of solutions on a tree structure we call the \textit{Ansatz tree}. We use the structure of $A$ in Assumption~\ref{assumeAUnitary} to construct such a tree based on the unitaries that make up $A$.
The core idea is to associate the nodes of this tree with the subspace states $\ket {u_i}$ from the previous section.
We show that a near-optimal solution is guaranteed to be found after enough nodes are included to the Ansatz tree. While the size of the tree may be very large in the worst case, the tree allows the systematic use of heuristic approaches to explore, prune, and expand it.

The construction of the full Ansatz tree is given in Definition~\ref{def:ansatztree}.
We start with the quantum state $\ket{b}$ and recursively construct the child nodes generated by the matrix $A$.
An illustration is shown in Figure~\ref{fig:Ansatztree}.
\begin{defn}
\label{def:ansatztree}
Given Assumption~\ref{assumeAUnitary}, $A = \sum_{k=1}^{K_A} \beta_k U_k$, we define the Ansatz tree recursively.
$$\left\{
	\begin{array}{ll}
		\mbox{The root of the tree is } \ket{b}. \\
		\mbox{Each node } \ket{\psi} \mbox{ on the tree has } $K$ \mbox{ child nodes: } U_1 \ket{\psi}, \ldots, U_K \ket{\psi}.
	\end{array}
\right.
$$
\end{defn}
Most straightforwardly one can take all the nodes of the Ansatz tree up to some depth as the solution subspace. We now derive guarantees for this approach which also show that the number of required nodes may be very large. We then discuss a heuristic approach to prune the tree in the less important directions and expand the tree in the more important directions, a method we call \textit{gradient expansion heuristics}. This heuristic can  reduce the number of nodes included in the subspace for a good solution.

\begin{figure}[t]
    \centering
    \includegraphics[width=0.9\textwidth]{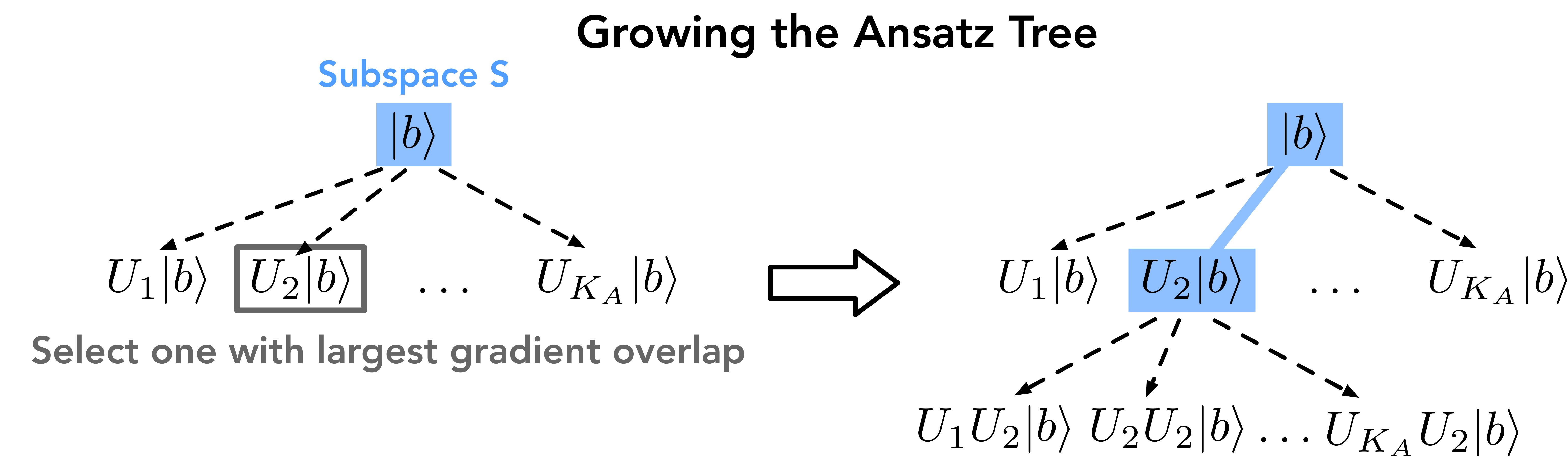}
    \caption{Illustration of the Ansatz tree and the gradient expansion strategy. The green region is the current subspace. On the left-hand side, we select a node that is a child of the nodes in the green region and that has the largest overlap with the gradient, here $U_2\ket{b}$. The right-hand side shows the tree after adding the new element to the subspace.}
    \label{fig:Ansatztree}
\end{figure}

Taking the full Ansatz tree and the regression loss from Definition \ref{defLossReg}, we are guaranteed to find a near-optimal solution after enough depth, see the next proposition.
\begin{prop}
\label{prop:Ansatzfullproof}
Fix $\epsilon > 0$ and let $A = \sum_{k=1}^{K_A} \beta_k U_k$ with $\rho(A) \leq 1, \rho(A^{-1}) \leq \kappa$. By selecting all nodes $\{\ket{u_1}, \ldots, \ket{u_m}\}$ on the Ansatz tree with depth at most $O(\kappa \log(\kappa / \epsilon))$, we have
$$\min_{\alpha_1, \ldots, \alpha_m \in \mathbb{R}}  L_R \left ( \sum_{i=1}^m \alpha_i \ket{u_i} \right) \leq \min_{x \in \mathbb{C}^{2^n}} L_R(x) + \epsilon.$$
\end{prop}
The result is an extension and variation of known results on using polynomial approximation of $1/x$ to solve linear systems of equations \cite{childs2017quantum}. See Appendix~\ref{sec:adapproof} for a detailed proof.
The required depth can be large in the worst case, especially when $\kappa$ is large. Moreover, this guarantee requires a large number of nodes on the Ansatz tree for a good approximation to the solution. In particular, the number of nodes scales exponentially in the condition number, i.e., $m = K_A^{\Ord{\kappa \log{(\kappa/\epsilon)}}}$.
 It is hence important to find ways to reduce the number of nodes.

One possible way of reducing the number of nodes is by solving a \textit{regularized} linear system of equations, using the loss function from Definition \ref{defLossTik}. Here, a polynomial number of nodes is enough to guarantee the performance of the solution even in the worst case, see the next proposition.
\begin{prop}
\label{prop:Ansatzproof}
Fix $\epsilon > 0$ and let $A = \sum_{k=1}^{K_A} \beta_k U_k$ with $\rho(A) \leq 1$. By selecting all nodes $\{\ket {u_1}, \ldots, \ket {u_m}\}$ on the Ansatz tree with depth at most $\ceil{C \log(1 / 2\epsilon)}$, where $C = 1 / \log(1 / (2 - \sqrt{3})) \approx 0.76$, we have
$$\min_{\alpha_1, \ldots, \alpha_m \in \mathbb{R}} L_T\left( \sum_i \alpha_i \ket {u_i} \right ) \leq \min_{x \in \mathbb{C}^{2^n}}  L_T(x)  + \epsilon.$$
The depth only depends on how good the approximation is, which is characterized by $\epsilon$. For example, when $\epsilon = 0.02$, we only need depth at most $4$ and number of nodes $m \leq {K_A}^4$.
\end{prop}
We sketch an argument for the $\log(1/\epsilon)$ dependency here based on standard convex optimization results and refer to Appendix~\ref{sec:adapproof} for a careful analysis.
First, using $A$ is Hermitian, the gradient of the loss function is $\nabla L_T(x) = x + 2 A^\dagger A x - 2 A^\dagger \ket b = x + 2 A^2 x - 2 A \ket b$, and the Hessian is given by $\nabla^2 L_T(x) = \mathbbm 1 + 2 A^2$.
Hence we have the bounds $\mathbbm 1 \preccurlyeq \nabla^2 L_T(x) \preccurlyeq 3 \mathbbm 1$ for all $x$.
Performing one gradient descent steps at the $t$-th iteration yields the following relation
$x^{(t+1)} \leftarrow x^{(t)} - \eta (x^{(t)} + 2 A^2 x^{(t)} - 2 A \ket b)$.
The step size $\eta$ can be determined by exact line search \cite{nocedal2006numerical}.
Suppose we start from $x^{(0)} = \ket{b}$, then after $T$ iterations, the solution can be expressed as a polynomial $p(z)$ of degree $T^2$ of the matrix $A$ applied to $\ket b$, i.e., $x^{(T)} = p(A) \ket b$, with well-defined polynomial coefficients. This polynomial directly expresses how to linearly combine the nodes in the Ansatz tree up to depth $T^2$.
From standard convex analysis \cite{boyd2004convex}, for strongly convex functions, if
$$T \geq \frac{\log ((L_T(\ket b) - \min_x L_T(x) )/\epsilon)}{\log 1/3} = \Tht{\log(1/\epsilon)},$$
we have $L_T(x^{(T)}) - \min_x L_T(x) \leq \epsilon$. Hence, we have derived a required depth of $\Ord{\log^2(1/\epsilon)}$.  Using Newton's method this depth can be improved to $\Ord{\log(1/\epsilon)}$.
As we have to include number of nodes exponential in the depth, the prefactor in the $\log 1/\epsilon$ dependency is crucial for near-term quantum computing applications.
Proposition~\ref{prop:Ansatzproof} shows a favorable prefactor using a more intricate proof (see Appendix~\ref{sec:adapproof}).

Another line of idea for reducing the number of nodes is to investigate methods which judiciously include nodes on the tree and prune branches which are not essential to the problem.
We introduce a heuristic procedure for exploring the tree which we call the \textit{gradient expansion heuristics}.
The key idea is to use the gradient information of the current state to expand the state space.
At the start, let the subspace $S$ contain only the root of the Ansatz tree, i.e., $S = \{\ket{b}\}$. At each step, we perform the following steps.
\begin{enumerate}
\item Solve for the optimal ${x^S} = \sum_{\ket{\psi_i} \in S} \alpha^*_{i} \ket{\psi_i}$ by optimizing over the combination parameters $\alpha_1, \ldots, \alpha_m$ as discussed in Section~\ref{sec:optcombparam}.
\item For each quantum state $\ket{\psi}$ in the set of child nodes of the set $S$ on the Ansatz tree,
denoted as $\mathcal{C}(S)$,
compute the gradient overlap $\bra \psi \nabla L_R(x^S) = 2 \sum_{\ket{\psi_i} \in S} \alpha^*_{i} \bra \psi A^2 \ket{\psi_i} - 2 \bra \psi A \ket b$. This can be done by estimating $\braket \psi {\psi_i}, \bra \psi A \ket {\psi_i}, \bra \psi A^2 \ket {\psi_i}$ for all $\ket {\psi_i} \in S$.
The overlaps can be computed efficiently using the Hadamard test via Proposition \ref{propHadamard} and \ref{propmodiHadamard}.
\item Add a new node to the subspace $S$, such that the node has the largest overlap with the gradient. More formally,  select $\ket{\psi^*} = {\rm arg} \max_{\ket \psi \in \mathcal{C}(S)} |\bra \psi \nabla L_R(x^S)|$ and grow the set $S \gets S \cup \{ \ket{\psi^*} \}$.
\end{enumerate}
An illustration of this method can be found in Figure~\ref{fig:Ansatztree}.
This gradient expansion procedure can be justified by the following proposition (See Appendix~\ref{sec:adapproof} for proof).
\begin{prop}
\label{prop:treegradient}
If $\ket{\psi^*}$ has gradient overlap $g = |\bra{\psi^*} \nabla L_R(x^S)| > 0$, then after expanding the subspace $S \leftarrow S \cup \{\ket{\psi^*}\}$ and optimizing the combination parameters, it is guaranteed that
$$L_R\left(x^{S \cup \{\ket{\psi^*}\}}\right) \leq L_R\left(x^S\right) - \frac{g^2}{4}.$$
\end{prop}
As a result of this proposition, if we find a new quantum state $\ket{\psi^*}$ with a nonzero gradient overlap, we are guaranteed that the next state vector $x^{S \cup \{\ket{\psi^*}\}}$ will be better than the current vector $x^S$.
Furthermore, a larger gradient overlap guarantees a larger decrease in the loss function.
Hence it is best to find a state vector that has the largest gradient overlap.

We now pinpoint limitations of the proposed Ansatz tree CQS approach.
The potential problems for variational algorithms discussed in Section~\ref{sec:potprob} do not apply to this Ansatz tree CQS approach and the problematic linear systems could actually be solved easily.
For the Tikhonov loss function, the CQS approach is also guaranteed to find a solution with near-optimal loss function in polynomial time.
However, for the regression loss function, this approach is guaranteed to efficiently find the solution \textit{only when} the condition number of $A$ is bounded by a constant.
Note that traditional quantum algorithms for linear system \cite{childs2017quantum, subacsi2019quantum} are efficient for condition numbers that are polynomial in the number of qubits. 
When the condition number of $A$ is too large, this Ansatz tree CQS approach cannot guarantee to find the optimal solution in polynomial time.
The gradient expansion heuristics can ameliorate this shortcoming in practice, but, in the worst case, it may still require an exponential amount of time when $\kappa$ is too large.

\subsection{Numerical experiments}
\label{sec:numerics}
\begin{figure}[t]
    \centering
    \includegraphics[width=1.0\textwidth]{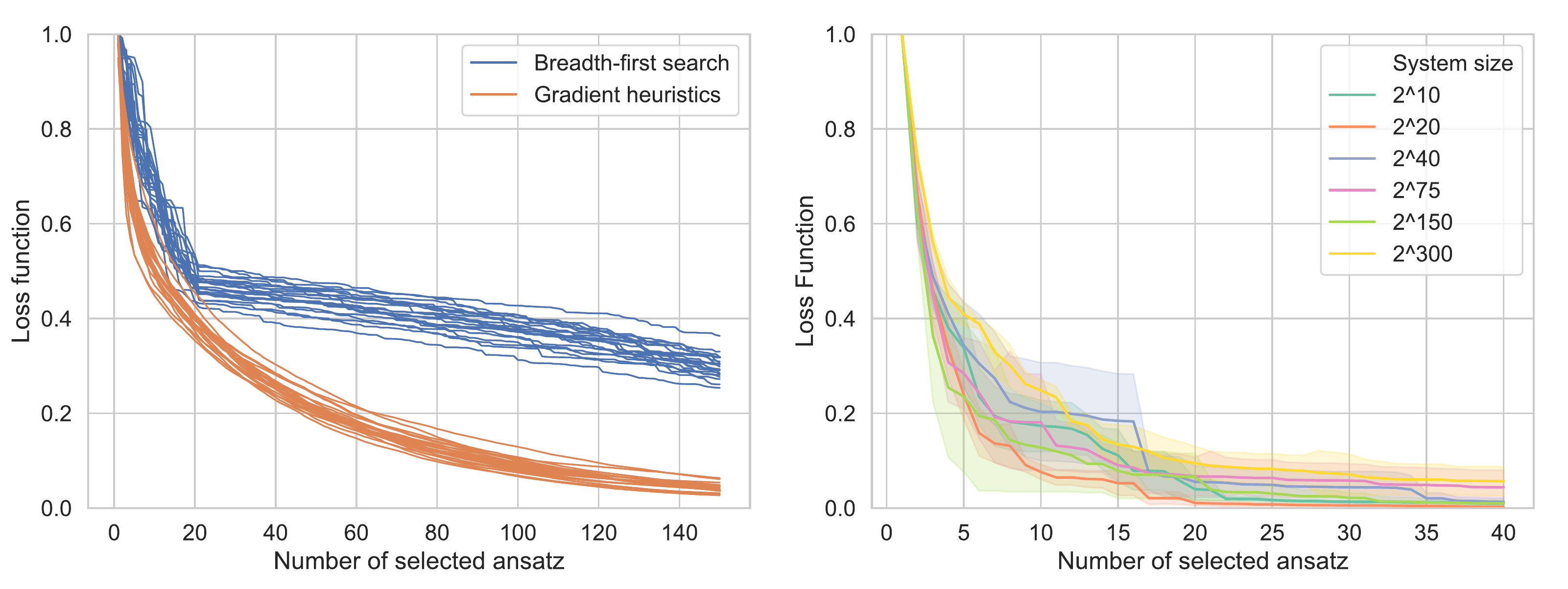}
    \caption{Numerical experiments on classical combination of quantum states algorithm for solving linear systems of equations. We use the standard loss function $L_R(x) = \norm{Ax - b}_2^2$.
    Left: Comparison of breadth-first search and the gradient expansion heuristics for adding nodes on the Ansatz tree to the subspace. We consider solving linear systems with a system size of $256 \times 256$ where $A$ is generated by sampling random weighted sum of Haar-random unitaries. Matrices $A$ generated this way have a large condition number (as large as the system size), so a large number of Ansatz states are needed find the solution. Each line represents an independent run.
    Right: Solving linear systems over a wide range of system sizes (from $2^{10} \times 2^{10}$ to $2^{300} \times 2^{300}$). For efficient classical simulation, the linear systems are generated as random weighted sums of Pauli strings, i.e., tensor product of Pauli operators. The shaded areas represent the standard deviation over five independent runs. }
    \label{fig:adapnumeric}
\end{figure}

We now present numerical experiments for the CQS-based algorithm.
The experiments are shown in Figure~\ref{fig:adapnumeric}.
In Figure~\ref{fig:adapnumeric} (left), we compare the use of gradient expansion heuristics with the use of a breadth-first search that simply includes every node on the Ansatz tree layer-by-layer.
We consider randomly generated linear systems of size $256 \times 256$.
We generate a random linear system by selecting several unitary matrices $U_1, \ldots, U_S$ from the Haar measure, random scalars $\alpha_1, \ldots, \alpha_S$ from uniform distribution $[-2, 2]$ and let $A = \sum_{i = 1}^S \alpha_i (U_i + U_i^\dagger)$. This construction guarantees that $A$ is Hermitian and is a weighted sum of unitary matrices.
The condition number generated this way is very large (in the order of the system size).
In particular, we consider $S = 10$ (hence $A$ is a sum of $20$ unitaries).
A clear improvement can be seen when using the gradient expansion heuristics. This gradient heuristics converges quickly to the optimal point.
On the other hand, a breath-first search, which layer by layer includes every node on the Ansatz tree, results in a very slow convergence after an initial rapid convergence for $20$ rounds (this includes the first layer of the Ansatz tree).

In Figure~\ref{fig:adapnumeric} (right), we consider a special class of (sparse) linear systems that are extremely large.
In particular, we consider system sizes ranging from $2^{10} \times 2^{10}$ to $2^{300} \times 2^{300}$ to investigate whether our new approach suffers from the plateau issue discussed in Section~\ref{sec:potprob}.
To facilitate classical simulation, we consider $A \in \mathbb{C}^{2^n \times 2^n}$ with efficient Pauli decomposition, i.e., $A = \sum_{i=1}^S \alpha_i P^{(i)}_1 \otimes \ldots \otimes P^{(i)}_n$, where $P^{(i)}_j$ is a single-qubit Pauli operator (including the identity).
In addition, we set $\ket{b} = \ket{0^n}$.
To generate a random matrix $A$, we sample each $\alpha_i$ from the uniform distribution over $[-2, 2]$, and a random tensor product of Pauli operators from the uniform distribution over $4^n$ possible choices. Here, we consider $S=8$.
Note that when $S=1$ and we only sample from $I$ or $X$ in the Pauli string, then we recover the toy problem that leads to the plateau issues discussed in Section~\ref{sec:potprob}.
From Figure~\ref{fig:adapnumeric}, we can see that this approach circumvents the plateau issue and has a clear convergence over an remarkably wide range of system sizes.

\subsection{Hamiltonian CQS approach and connection to previous results}

The Ansatz tree approach so far is based on the direct use of the decomposition of the matrix $A$, see Assumption \ref{assumeAUnitary}. As an alternative approach, one can generate a subspace by performing Hamiltonian simulation. Because of the use of Hamiltonian simulation, this approach is less near-term compared to using the Ansatz tree. The benefit of this approach is that it greatly reduces the size of the Ansatz subspace compared to Proposition~\ref{prop:Ansatzfullproof}.

Ref.~\cite{childs2017quantum} shows that the function $1/x$ can be well approximated by a Fourier series with near-optimal number of terms. This Fourier series can be viewed as a series for approximating $1/A$ and obtains a linear combination of unitaries of the form $e^{-iAt_j}$ with pre-specified times $t_j$. Using a truncated Taylor method for quantum simulation of this linear combination \cite{berry2015simulating} obtains a quantum state proportional to the exact solution $A^{-1} \ket{b}$. Based on Ref.~\cite{childs2017quantum} (Lemma 11), we can choose the set of Ansatz states as follows
\be \label{eqAnsatzHS}
\Bigg\{ \ket{u_j} := e^{-iAt_j}\ket{b} \,\,\, \Bigg| \,\,\,  t_j = \frac{\epsilon j}{\kappa \log(\kappa / \epsilon)}, \,\, j = - J, \ldots, J \Bigg\},
\ee
where $J= \Tht{\kappa^2 \log^2(\kappa / \epsilon) / \epsilon }$ and $\kappa$ is an upper-bound on the condition number of $A$.
This set is of size $2J +1 = \Tht{\kappa^2 \log^2(\kappa / \epsilon) / \epsilon}$ and the number of gates to generate each Ansatz is about $t_J = \Ord{\kappa \log(\kappa / \epsilon)}$ (assuming each application of $A$ takes constant number of gates).

Using the CQS strategy, we are guaranteed to find an $\epsilon$-close solution for $Ax = \ket{b}$ because of the Fourier approximation results of \cite{childs2017quantum}.
As mentioned, the single run circuit depth is $O(\kappa \log(\kappa / \epsilon))$ and uses only one additional ancilla for performing the Hadamard test.
Because we never create the solution $x$ on a quantum computer, but emulate $x$ quantum-classically, we avoid the need of many ancilla qubits in existing quantum algorithms based on amplitude amplification and function approximations \cite{berry2015simulating, childs2017quantum}.
We now compare this approach to the random-sampling approach taken in \cite{subacsi2019quantum}.

The approach in \cite{subacsi2019quantum} eliminates the use of amplitude amplification and function approximation.
In that work, the operators $e^{-i H(s_j) t_j}$ are used with the Hamiltonian from Definition \ref{defHamiltonianA}. The  $s_j$ are fixed via a natural parameterization and the $t_j$ are randomly sampled from an interval of size about $[0,\Ord{\kappa^2}]$ for $j = 1,\dots,\Tht{\log^2(\kappa)/\epsilon}$.
The method is able to again achieve an $\epsilon$-close approximation.
The circuit depth is about $\Ord{\kappa \log(\kappa) /\epsilon}$, not counting details of the Hamiltonian simulation for $H(s)$.
In comparison, the Ansatz defined by the set in Eq.~(\ref{eqAnsatzHS}) shows the single-run circuit depth of $\Ord{\kappa \log(\kappa / \epsilon)}$.
As an example, when $\epsilon = 0.01$, we would achieve a roughly $100$-times reduction in the circuit depth.
This reduction in quantum resources comes at the cost of more classical repetitions. Hence, this Hamiltonian CQS approach fits into the general near-term strategy of trading the more expensive circuit depth for the cheaper number of runs of the experiments.

\section{Discussion}

The work provides algorithms for solving linear systems on near-term quantum computers.
The flavor of the presented algorithms is two-fold. The first set of algorithms are variational in nature and draw their inspiration from other variational quantum algorithms for quantum chemistry \cite{peruzzo2014variational, wecker2015progress, o2016scalable, colless2018computation} and quantum optimization \cite{farhi2014quantum, moll2018quantum, wang2018quantum}. For such algorithms, the quantum computer implements a single wavefunction Ansatz which is dependent on a set of variational parameters, usually in a non-linear fashion. The type of Ansatz in this setting can, in the extreme cases, be linear system-independent (agnostic) or fully linear system-dependent. As the agnostic case is useful for example when limitations of the hardware are dominating the overall implementation, such Ans\"atze have also been called ``hardware efficient" \cite{kandala2017hardware}.
On the other hand, the dependent Ansatz takes into account the linear system at the cost of requiring Hamiltonian simulation which increases the overall complexity of implementing such an Ansatz in a near-term quantum processor.
In this work, we exhibit types of linear systems for which variational approaches with a polynomial number of variational parameters show plateau issues. These plateau issues may however not arise for other types of linear systems and Ans\"atze with more structure and when using different loss functions. An important future work is to better characterize those linear systems where variational methods offer near-term quantum advantages.

The second set of approaches are based on classical combination of variational quantum states (CQS).
The method is inspired from the basic concept of diversification and robustness. Using a single class of methods can provide only limited benefits when compared to combining multiple different methods and using the best parts of each.
This method introduces a new set of combination parameters to add together different variational quantum states. The combination is emulated classically rather than represented directly on the quantum computer. Hence, the method increases the overall expressiveness and power of the Ansatz without the need of additional quantum resources.
We can use the variational states such as the ones presented in the first part, and also others yet to be developed.
Our CQS approach is also reminiscent of techniques used for example in quantum chemistry, where the linear combination of atomic orbitals (LCAO) approach allows to optimally construct molecular orbitals from atomic orbitals.

To avoid the complexity of optimizing the variational parameters, which can involve an ill-shaped optimization landscape with plateaus and local minima, we have proposed an approach that alternates between solving for the optimal solution in a subspace and growing the subspace on an Ansatz tree.
This approach is inspired by the Krylov subspace method in solving linear systems. Krylov subspace is a subspace spanned by $\{b, Ab, A^2b, \ldots, A^{r-1} b\}$, which is similar to the Ansatz tree we defined. The Krylov subspace method solves for the optimal solution in the subspace and increases $r$ if the obtained solution is not good enough. A popular variant of Krylov subspace method for linear systems is the conjugate gradient method \cite{shewchuk1994introduction}.
The Ansatz tree is also reminiscent of the coupled-cluster Ansatz in quantum chemistry \cite{vcivzek1966correlation, monkhorst1977calculation, purvis1982full, bartlett1989coupled}, which systematically takes into account higher and higher orders of the electron correlation at the cost of increasing the complexity of preparing the Ansatz.

We have performed numerical experiments solving exponentially large linear systems with sizes up to $2^{300} \times 2^{300}$.
These experiments are achieved by considering a special class of linear systems that allows efficient simulation of the proposed quantum algorithm on a classical computer.
To achieve actual quantum advantage, we require either $A$ to be a sum of unitaries that cannot be simulated efficiently on a classical computer or $b$ to be a quantum state generated by some quantum circuit.
It should be noted that there will always be a trade-off between how near-term the quantum algorithm is (the required quantum coherence, entanglement, and interference) and
how much quantum advantage we can expect from executing the quantum algorithm.
An important future direction would be a detailed analysis on the performance of the proposed algorithms under the effect of decoherence and imperfections of real-world quantum devices.
We believe the synthesis and future improvement of the proposed ideas can provide real benefits for solving linear systems when quantum computers achieve sizes of $50$-$70$ high-quality qubits.

\section{Acknowledgements}
We would like to thank Fernando Brandao, Yudong Cao, Richard Kueng, John Preskill, Ansis Rosmanis, Miklos Santha, Thomas Vidick, and Zhikuan Zhao for valuable discussions.
H.H.~is supported by the Kortschak Scholars Program and thanks the hospitality of the Centre for Quantum Technologies.
K.B.~acknowledges the CQT Graduate Scholarship.
P.R.~acknowledges support from Singapore's Ministry of Education and National Research Foundation and Baidu.
\bibliographystyle{apsrev}
\bibliography{VQE}

\appendix

\section{Measurements}
\label{secMeasure}

 \begin{lemma} \label{lemmaPauliMeasurement}
 Let $\epsilon>0$ and $P_k$ be a certain Pauli string over $n$ qubits. Let multiple copies of an arbitrary $n$-qubit quantum state $\ket \psi$ be given. The expectation value $\bra \psi P_k \ket \psi$ can be determined to additive accuracy $\epsilon$ with failure probability at most $\delta$ using $\Ord{\frac{ 1 }{\epsilon^2} \log (\frac{1}{\delta} )}$  copies of $\ket \psi$.
 \end{lemma}
 \begin{proof}
 A single measurement obtains the outcome $m_\pm = \pm 1$. We have $\bra \psi P_k \ket \psi = p m_+  + (1-p) m_- = 2 p -1$, where $p$ is the probability of measuring $+1$. To estimate this probability, perform independent trials of the Bernoulli test. Each trial has expectation value $p$. We use the statistic $M_+/M$, where $M_+$ is the number of positive outcomes over $M$ trials.
For the error estimate, we require
$P\left [ \vert 2 M_+/M -1 - \bra \psi P_k \ket \psi  \vert \geq \epsilon  \right] \leq \delta$
from which the number of measurements is
$\Ord{\frac{ 1 }{\epsilon^2} \log (\frac{1}{\delta} )}$ via Hoeffding's inequality. 
 \end{proof}
\begin{prop}[Swap test]
Given multiple copies of $n$-qubit quantum states $\ket{u}$ and $\ket{v}$.
There is a quantum algorithm that determines the overlap $ \vert \braket{v}{ u} \vert ^2$ to additive accuracy $\epsilon$ with failure probability at most $\delta$ using $\Ord{\frac{ 1 }{\epsilon^2} \log (\frac{1}{\delta} )}$ copies and $\tOrd{\frac{ 1 }{\epsilon^2} \log (\frac{1}{\delta} )}$ operations.
\end{prop}
\begin{proof}
Use an ancilla and perform a controlled swap
$\frac{1}{\sqrt{2}} \left ( \ket 0 + \ket 1 \right ) \ket{u}\ket{v} \to \frac{1}{\sqrt{2}} \left ( \ket 0 \ket{u}\ket{v}  +  \ket 1 \ket{v}  \ket{u}\right )$.
Performing a Hadamard on the ancilla obtains $\frac{1}{2} \left ( \ket 0 (\ket{u}\ket{v}  + \ket{v}\ket{u} ) +   \ket 1 ( \ket{u}\ket{v}  - \ket{v}\ket{u}  )\right ) =: \ket{\xi}$.
Now measure the ancilla in $Z$. The expectation value is
$\bra{\xi} Z\ket{\xi} = \frac{1}{4} \left ( \bra{u}\bra{v}  + \bra{v}\bra{u} \right ) \left (\ket{u}\ket{v}  + \ket{v}\ket{u} \right ) -  \frac{1}{4} \left ( \bra{u}\bra{v}  - \bra{v}\bra{u} \right ) \left (\ket{u}\ket{v}  - \ket{v}\ket{u} \right )=   \vert \braket{v} {u} \vert^2$.
\end{proof}
We can also measure the real and imaginary part separately under a different input model.
\begin{prop}[Hadamard test] \label{propHadamard}
Assume the controlled state preparation
$U_{\rm prep} = \ket 0 \bra 0  \otimes U_{v_0} +  \ket 1  \bra 1 \otimes U_{v_1}$,
with
$U_{v_j} \ket{ 0^n} = \ket{v_j}$.
There is a quantum algorithm that determines $\realpar{ \braket{v_0}{v_1}}$ and $\impar{ \braket{v_0}{v_1}}$ to additive accuracy $\epsilon$ with failure probability at most $\delta$ using $\Ord{\frac{ 1 }{\epsilon^2} \log (\frac{1}{\delta} )}$ applications of $U_{\rm prep}$ and $\tOrd{\frac{ 1 }{\epsilon^2} \log (\frac{1}{\delta} )}$ operations.
\end{prop}
\begin{proof}
Use an ancilla prepared in $( \ket 0 + \alpha \ket 1 )/\sqrt{2}$ with $\alpha=1$ or $\alpha=i$. Apply $U_{\rm prep}$ to obtain
$\frac{1}{\sqrt{2}} \left ( \ket 0 \ket{v_0} + \alpha \ket 1  \ket {v_1} \right )$. Another Hadamard on the ancilla obtains
$\frac{1}{2} \left ( \ket 0 ( \ket {v_0} + \alpha  \ket{ v_1}) +   \ket 1 ( \ket{v_0}-\alpha \ket{v_1} )\right )=:\ket{\xi}$.
Now measure the ancilla in $Z$. The expectation value is
$\bra \xi  Z \ket \xi =\frac{1}{2}  \left (  \alpha  \braket{ v_0}{ v_1} + \alpha^\ast  \braket{ v_1}{ v_0} \right)$. If $\alpha =1$, then  $\bra \xi  Z  \ket \xi= \realpar{\braket{ v_0}{ v_1}}$. If $\alpha=i$, then
$\bra \xi  Z \ket \xi  = \impar{\braket{ v_0}{ v_1}}$.
\end{proof}

\begin{prop}[Measuring the Hamiltonian Loss Function]
Given Assumptions \ref{assumeAUnitary} on the unitary decomposition of $A$ and multiple copies of the quantum state $\ket x$. The loss function $\bra x  A^2 \ket x - \bra x  A \ket{b}\bra{b} A \ket x$  can be estimated efficiently on a quantum computer. More precisely, an estimate up to additive error $\epsilon$ can be obtained with probability $1-\delta$ using $\Ord{\frac{(\sum_{k} |\beta_{k}|)^4 }{ \epsilon^2} \log \frac{K_A}{\delta}}$ quantum measurements.
\end{prop}
\begin{proof}

From Definition \ref{assumeAUnitary}, define the size parameter $\eta = \sum_{k, l} { \vert\beta_{k} \beta_{l}\vert}$.
The first term of the loss function is $A^2 = \sum_{k} \sum_{l} \beta_{k} \beta_l U_k U_l$.
Measure each term $\bra x U_k U_l \ket x$ individually using the Hadamard test in Proposition~\ref{propHadamard} with $\ket{v_0} = \ket{x}$ and $\ket{v_1} = U_k U_l \ket{x}$.
For each term, perform $\Ord{\frac{  \vert\beta_{k} \beta_{l}\vert \eta }{  \epsilon^2} \log \frac{K_A}{\delta} }$ quantum measurements to produce an estimate $\widetilde{\bra x U_k U_l \ket x}$ of $\bra x U_k U_l \ket x$ up to additive accuracy $  \epsilon / \sqrt{\eta \vert\beta_{k} \beta_{l}\vert  } $ with success probability $1-\frac{\delta}{2 K_A^2}$.
Thus the additive accuracy for the estimate $\widetilde{\bra x  A^2 \ket x}$ of $\bra x  A^2 \ket x$ is $\left \vert \widetilde{\bra x  A^2 \ket x} - \bra x  A^2 \ket x\right \vert
\leq\epsilon$, which is obtained by adding up the variances of each independent estimate scaled by the coefficients $\vert\beta_{k} \beta_{l}\vert^2$ and taking the square root. The total success probability is $\left (1-\frac{\delta}{2 K_A^2}\right)^{K_A^2} \geq e^{-\delta}\geq 1-\delta$.
 The total number of measurements is $\Ord{\frac{\sum_{kl} \vert \beta_k \beta_l \vert \eta}{ \epsilon^2} \log \frac{K_A}{\delta} } = \Ord{\frac{(\sum_{k} \vert \beta_{k}\vert)^4}{ \epsilon^2}  \log \frac{K_A}{\delta}}$.

For the second term, with $A = \sum_k \beta_k U_k$ estimate $\bra{b}A\ket{x} =  \sum_k \beta_k \bra{b} U_k \ket{x}$. Use the Hadamard test in Proposition~\ref{propHadamard} with $\ket{v_0} = \ket{b}$ and $\ket{v_1} = U_k \ket{x}$ to estimate each term $\bra{b} U_k \ket{x}$.
For each term, perform $\Ord{ \frac{\vert \beta_k \vert \left(\sum_{k'} \vert \beta_{k'} \vert\right)^3 }{ \epsilon^2} \log \frac{K_A}{\delta}}$ quantum measurements to estimate $\bra b U_k \ket x$ up to variance $\epsilon^2 /\left (\vert \beta_k \vert \left(\sum_{k'} \vert \beta_{k'}\vert \right)^3 \right)$ with success probability $1-\frac{\delta}{8K_A}$.
Thus the variance in the estimation of $\bra{b}A\ket{x}$ is $\epsilon^2 / (\sum_{k'} \vert \beta_{k'}\vert )^2$ with success probability $1-\frac{\delta}{4}$ using a total of $\Ord{\frac{\left(\sum_{k} \vert \beta_{k}\vert \right)^4}{ \epsilon^2} \log \frac{K_A}{\delta}}$ quantum measurements.
The term $\bra x  A \ket{b}\bra{b} A \ket x$ can be estimated by performing two independent estimations of $\bra x  A \ket{b}$ and $\bra{b} A \ket x$ and multiplying them together.
The variance of $\bra x  A \ket{b} \bra{b} A \ket x$ is $2 \vert \bra x  A \ket{b}\vert^2 \Var[\bra{b} A \ket x]+
\Var[\bra{b} A \ket x]^2$ which is bounded by $4 \vert \bra x  A \ket{b}\vert^2 \Var[\bra{b} A \ket x] \leq 4 \left(\sum_{k} \vert \beta_{k}\vert \right)^2 \cdot \epsilon^2 / \left(\sum_{k} \vert \beta_{k}\vert \right)^2 = 4 \epsilon^2$ with success probability $1-\delta$ and uses a total of $\Ord{\frac{\left(\sum_{k} \vert \beta_{k}\vert \right)^4 }{\epsilon^2} \log \frac{K_A}{\delta}}$ quantum measurements.

Summing over the first and the second term obtains an estimation for the loss function with variance $5 \epsilon^2$.
By considering $\epsilon \leftarrow \epsilon / \sqrt{5}$ and $\delta \leftarrow \delta/4$, an estimation for the loss function with variance $\epsilon^2$ and success probability $1-\delta$  is obtained using $\Ord{\frac{\left(\sum_{k} \vert \beta_{k}\vert \right)^4}{ \epsilon^2} \log\frac{K_A}{\delta}}$ quantum measurements.
\end{proof}

\begin{prop}[Modified Hadamard test] \label{propmodiHadamard}
Assume the controlled state preparation
$U_{\rm prep} = \ket 0 \bra 0 \otimes U_{v_0} +  \ket 1 \bra 1 \otimes U_{v_1}$,
with
$U_{v_j} \ket{ 0^n} = \ket{v_j}$ is an $n$-qubit state.
Given an observable $O = U^\dagger D U \in \mathbb{C}^{2^n \times 2^n}$, where $U$ is a unitary matrix that can be implemented efficiently as a quantum circuit, $D$ is a real diagonal matrix, and $D_{ii}$ can be computed efficiently as a classical function $f: 2^n \rightarrow [-1,1]$.
Both $\realpar{ \bra{v_0}O\ket{v_1}}$ and $\impar{ \bra{v_0}O\ket{v_1}}$ can be estimated efficiently on a quantum computer.
More precisely, we can estimate $\realpar{ \bra{v_0}O\ket{v_1}}$ and $\impar{ \bra{v_0}O\ket{v_1}}$ to additive accuracy $\epsilon$ with failure probability at most $\delta$ using $\Ord{\frac{1}{\epsilon^2} \log(\frac{1}{\delta})}$ quantum measurements.
\end{prop}
\begin{proof}
Use an ancilla prepared in $\left( \ket 0 + \alpha \ket 1 \right)/\sqrt{2}$ with $\alpha=1$ or $\alpha=i$. Apply $U U_{\rm prep}$ to obtain
$\frac{1}{\sqrt{2}} \left ( \ket 0 \otimes U \ket{v_0} + \alpha \ket 1 \otimes U \right \ket{v_1} )$. Another Hadamard on the ancilla gives
$\ket{\xi} = \frac{1}{2} (I \otimes U) \left ( \ket 0 (\ket{v_0} + \alpha  \ket{v_1}) +   \ket 1 ( \ket{v_0} -\alpha \ket{v_1} )\right )$.
Now measure the ancilla in $Z$ to obtain $z_a \in \{\pm 1\}$ and the rest of the state in the computational basis to obtain $b \in \{0, 1\}^n$.
Then compute $z_a f(b)$.
The expectation value is $\bra{\xi} Z\otimes D \ket{\xi} = \frac{1}{2}  \left (  \alpha  \bra{ v_0} O \ket{ v_1} + \alpha^\ast  \bra{ v_1} O \ket{ v_0} \right)$. If $\alpha = 1$, then $\bra{\xi} Z\otimes D \ket{\xi} = \realpar{ \bra{v_0}O\ket{v_1}}$. If $\alpha = i$, then $\bra{\xi} Z\otimes D \ket{\xi} = {\rm Im} \bra{v_0}O\ket{v_1}$.
Because each sample $z_a f(b)$ is between $-1$ and $1$, Hoeffding's inequality shows that taking the average of $\Ord{\frac{1}{\epsilon^2} \log(\frac{1}{\delta})}$ samples results in an estimate within error at most $\epsilon$ with probability at least $1 - \frac{1}{\delta}$.
\end{proof}

\section{Potential problems for local loss function in the presence of entangling gates}
\label{app:moreprob}

Consider the case where $A \in \mathbb{C}^{2^n \times 2^n}$ is Hermitian and unitary (equivalently, all eigenvalue of $A$ are $1$ or $-1$) and $\ket{b} = U_b \ket{0^n}$.
And let $\ket{x(\theta)}$ be some pre-specified Ansatz with variational parameter $\theta$.
We now analyze the local loss function
$$L_L(\theta) = \bra{x(\theta)} A U_b \left(\mathbbm 1 - \frac{1}{n} \sum_{i=1}^n \ket{0_i}\bra{0_i}\right) U_b^\dagger A \ket{x(\theta)} = 1 - \frac{1}{n} \sum_{i=1}^n \bra{x(\theta)} A U_b \ket{0_i}\bra{0_i} U_b^\dagger A \ket{x(\theta)}.$$
When $U_b$ is a random quantum circuit consisting of $\Ord{n^2}$ 1D nearest neighbor two-qubit gates, \cite{brandao2016local} shows that for any quantum state $\ket{x(\theta)}$, the two-design property of the random unitary $U_b^\dagger$ gives
\be \label{eq:tdesign}
\norm{\mathbb{E}_{U_b, A} \left[(U_b^\dagger A \ket{x(\theta)} \bra{x(\theta)} A U_b)^{\otimes 2}\right] - \frac{\mathbbm 1+S}{(2^n + 1)2^n}}_1 \leq \frac{1}{2^{n^2}},
\ee
where $S$ is the swap operator and $\norm{\cdot}_1$ is the trace norm.
Let $R_i = \bra{x(\theta)} A U_b \ket{0_i}\bra{0_i} U_b^\dagger A \ket{x(\theta)}$ be a random variable.
Then using $\frac{1}{2} = \frac{1}{2} \bra{x(\theta)} A U_b (\ket{0_i}\bra{0_i} + \ket{1_i}\bra{1_i}) U_b^\dagger A \ket{x(\theta)}$, we have
$$R_i - \frac{1}{2} = \frac{1}{2} \bra{x(\theta)} A U_b (\ket{0_i}\bra{0_i} - \ket{1_i}\bra{1_i}) U_b^\dagger A \ket{x(\theta)} = \frac{1}{2} \bra{x(\theta)} A U_b Z_i U_b^\dagger A \ket{x(\theta)},$$
where $Z_i$ is the Pauli Z operator on the $i$-th qubit.
Hence by Markov's inequality,
\be \label{eq:markov}
\mathbb{P}\left[ \left|R_i - \frac{1}{2}\right| > \frac{1}{2^{n/4}} \right] \leq 2^{n/2} \mathbb{E}\left[\left|R_i - \frac{1}{2}\right|^2\right] = \frac{1}{4} 2^{n/2} \Tr\left( Z_i^{\otimes 2} \,\, \mathbb{E}_{U_b, A} \left[(U_b^\dagger A \ket{x(\theta)} \bra{x(\theta)} A U_b)^{\otimes 2}\right] \right).
\ee
For any linear operator $P, Q$ and $Q'$, we have
\be \label{eq:basicla}
|\Tr(PQ) - \Tr(PQ')| \leq \norm{P}_\infty \norm{Q - Q'}_1 \implies \Tr(PQ) \leq \Tr(PQ') + \norm{P}_\infty \norm{Q - Q'}_1.
\ee
Combining Inequality~\eqref{eq:tdesign} and \eqref{eq:basicla}, we have
$$\Tr\left( Z_i^{\otimes 2} \,\, \mathbb{E}_{U_b, A} \left[(U_b^\dagger A \ket{x(\theta)} \bra{x(\theta)} A U_b)^{\otimes 2}\right] \right) \leq \Tr\left( Z_i^{\otimes 2} \,\, \frac{\mathbbm 1+S}{(2^n + 1)2^n} \right) + \frac{1}{2^{n^2}} = \frac{1}{2^n} + \frac{1}{2^{n^2}}.$$
Now putting this result back to the Markov's inequality in \eqref{eq:markov}, we have
$$\mathbb{P}\left[ \left|R_i - \frac{1}{2}\right| > \frac{1}{2^{n/4}} \right] \leq \frac{1}{4}\left( \frac{1}{2^{n/2}} + \frac{1}{2^{n^2 - n/2}} \right) \leq \frac{1}{2^{n/2}}.$$
Using union bound and the definition of the local loss function, we have
$$\left|L_L(\theta) - \frac{1}{2}\right| \leq \frac{1}{2^{n/4}},$$
with probability at least $1 - n / 2^{n/2}$.
Hence, for any $\theta_1, \ldots, \theta_M$ (there can be an exponential number of them, i.e., $M \leq \frac{2^{n/2}}{50 n}$),
$$\left| L_L(\theta_i) - \frac{1}{2}\right| \leq \frac{1}{2^{n/4}}, \quad \forall i = 1, \ldots, M,$$
with probability at least $0.98$.
This means that even if we randomly select a very large number of variational parameters, the local loss function will still be exponentially close to the value $1/2$ for all of them.

Now, consider a local expansion of $L_L(\theta)$ around $\theta^0$,
$$1 - \frac{1}{n} \sum_{i=1}^n \left( \bra{x(\theta^0)} A U_b \ket{0_i} \bra{0_i} U_b^\dagger A \ket{x(\theta^0)} -  \sum_{k=1}^m 2 \realpar{\bra{x(\theta^0)} A U_b \ket{0_i} \bra{0_i} U_b^\dagger A \frac{\partial}{\partial \theta_k} \ket{x(\theta^0)}} \delta \theta_k \right) + \Ord{m^2 \delta \theta_i^2}.$$
We analyze the partial derivative of the loss function in $\theta_k$.
We consider the random variable $g_k = 2 \realpar{\bra{x(\theta^0)} A U_b \ket{0_i} \bra{0_i} U_b^\dagger A \frac{\partial}{\partial \theta_k} \ket{x(\theta^0)}}$, and define $u = \ket{x(\theta^0)}$ and $v = \frac{\partial}{\partial \theta_k} \ket{x(\theta^0)}$.
We assume that $\norm{v} = \norm{\frac{\partial}{\partial \theta_k} \ket{x(\theta^0)}} \leq C$, for some constant $C$, which is true when the variational parameters are single-qubit rotation angles.
The random variable $g_k$ could be written as
\be \label{eq:gradllf}
g_k = \Tr\left( \left(v u^\dagger + u v^\dagger\right) \,\, AU_b \ket{0_i} \bra{0_i} U_b^\dagger A\right).
\ee
We first note that the normalization condition of quantum states $\braket{x(\theta)}{x(\theta)} = 1$ implies that $\frac{\partial}{\partial \theta_k} \bra{x(\theta)} \ket{x(\theta)} = u^\dagger v + v^\dagger u = 0$.
Using the results in \cite{brandao2016local}, the two-design property of the random unitary $U_b$ gives
\be \label{eq:grad2design}
\norm{\mathbb{E}_{U_b, A} \left[(A U_b \ket{0_i} \bra{0_i} U_b^\dagger A)^{\otimes 2}\right] - \frac{(2^{2n - 2} - 2^{-1})\mathbbm 1 + (2^{n-1} - 2^{n-2})S}{2^{2n} - 1}}_1 \leq \frac{1}{2^{n^2}}.
\ee
Now we bound the second moment of the random variable $g_k$,
\begin{align*}
\mathbb{E}_{U_b, A} \left[ g_k^2 \right] & = \Tr\left( \left(v u^\dagger + u v^\dagger\right)^{\otimes 2} \,\, \mathbb{E}_{U_b, A} \left[ (AU_b \ket{0_i} \bra{0_i} U_b^\dagger A)^{\otimes 2} \right] \right) \\
& \leq \Tr\left( \left(v u^\dagger + u v^\dagger\right)^{\otimes 2} \,\, \frac{(2^{2n - 2} - 2^{-1})\mathbbm 1 + (2^{n-1} - 2^{n-2})S}{2^{2n} - 1} \right) + \frac{1}{2^{n^2}} \norm{\left(v u^\dagger + u v^\dagger\right)^{\otimes 2}}_\infty,\\
& = \frac{2^{n-1} - 2^{n-2}}{2^{2n} - 1} \Tr\left(\left(v u^\dagger + u v^\dagger\right)^{\otimes 2} S\right) + \frac{1}{2^{n^2}} \norm{\left(v u^\dagger + u v^\dagger\right)^{\otimes 2}}_\infty,\\
& \leq \frac{2^{n-1} - 2^{n-2}}{2^{2n} - 1} 4C^2 + \frac{1}{2^{n^2}} 4C^2 \leq \frac{8 C^2}{2^n}.
\end{align*}
The first line uses Equation~\eqref{eq:gradllf}, and the second line uses Inequality~\eqref{eq:basicla} and~\eqref{eq:grad2design}.
The third line uses the fact that $u^\dagger v + v^\dagger u = 0$, while the last line uses the fact that $\norm{v} \leq C$.
Again by Markov's inequality, we have
$$\mathbb{P} \left[|g_k| > \frac{2 C}{2^{n/4}} \right] \leq \frac{2^{n/2} \mathbb{E}_{U_b, A}\left[g_k^2\right]}{4 C^2} \leq \frac{2}{2^{n/2}}.$$
Thus given $m \ll 2^{n/2}$, union bound on the above result gives
$$|g_k| \leq \frac{2 C}{2^{n/4}}, \,\, \forall k = 1, \ldots, m,$$
with high probability (more precisely, with probability $\geq 1 - 2m / 2^{n/2}$).
This means that the gradient would be exponentially small and the optimization landscape is locally flat around $\theta^0$.
From the analysis, we can see that the local loss function $L_L(\theta)$ will plateau at the value $1/2$ with an exponentially small deviation at each point and an exponentially small gradient.

\section{Detailed analysis on optimizing combination coefficients}
\label{app:optcomb}

The necessary quantum measurements for obtaining the inputs $Q, r$ needed to solve the combination coefficients is discussed in the following lemma.
Note that here we are interested in a single sample of each $Q_{ij}$ and $r_i$, each of which is made up of a sum of terms $K_A^2$ and $K_A$, respectively.
Hence, each term in the sum is sampled once and the terms are added up to provide the single sample.
\begin{lemma} \label{lem:measurequad}
Given the unitary decomposition of matrix $A$ from Assumption \ref{assumeAUnitary}, the circuit for $\ket b$ via Assumption \ref{assumeB}, and let the circuit for creating $\ket {u_i}$ be $W_i$. Assume that controllable unitaries of all involved unitaries can be constructed. Then, we can obtain a single sample of
$(V^\dagger V)_{ij} = \bra {u_i} A^\dagger A \ket {u_j}$ and $q_i = \bra {u_i} A^\dagger \ket b$
from $O(K_A^2)$ and $\Ord{K_A}$ quantum measurements, where the magnitude of each sample is bounded by $\Ord{\left (\sum_{i=1}^{K_A} \vert \alpha_k \vert \right )^2}$ and $\Ord{\sum_{i=1}^{K_A} \vert \alpha_k \vert}$, respectively.
\end{lemma}
\begin{proof}
Note that $\bra {u_i} A^\dagger A \ket {u_j} = \sum_{k,k'=1}^{K_A} \alpha_k \alpha_{k'} \bra { 0^n} W_i^\dagger U_k^\dagger U_{k'} W_j \ket { 0^n}$. Construct the unitaries
$U_{prep,k,k'} = \ket 0 \bra 0 \otimes \mathbbm 1 +  \ket 1  \bra 1 \otimes W_i^\dagger U_k^\dagger U_{k'} W_j$.  Use the Hadamard test via Proposition \ref{propHadamard} to obtain a single estimate for $\bra { 0^n} W_i^\dagger U_k^\dagger U_{k'} W_j \ket { 0^n}$. The absolute value of the estimate is bounded by $\Ord{1}$ because the estimate for the real and imaginary part are both bounded by $1$.
We can estimate $\bra { 0^n} W_i^\dagger U_k^\dagger U_{k'} W_j \ket { 0^n}$ for all $k, k'$ using $\Ord{K_A^2}$ quantum measurements.
Then we obtained an estimate for $\realpar{\bra {u_i} A^\dagger A \ket {u_j}}$, where the absolute value is bounded by $\Ord{\left (\sum_{i=1}^{K_A} |\alpha_k|\right)^2}$.
Similar steps and constructing the unitaries
$U_{prep,k,b} = \ket 0 \bra 0 \otimes \mathbbm 1 +  \ket 1  \bra 1 \otimes W_i^\dagger U_k^\dagger U_b$ allow us to obtain an estimate for $\bra {w_i} A^\dagger \ket b$ with absolute value bounded by $\Ord{\sum_{i=1}^{K_A} \vert \alpha_k\vert}$ using $\Ord{K_A}$ quantum measurements.
\end{proof}

Because of the quantum measurements, we can only obtain an estimate of $Q$ and $r$.
The following proposition summarizes the required number of measurements to achieve a good solution using the approximations $\hat Q$ and $\hat r$ obtained from quantum measurement.
This proposition focuses on real optimization where the variable $z$ is real.
\begin{prop}
\label{prop:guaranteesubspace}
Consider a quadratic function $L(z) = z^T Q z - 2 r^T z $, where $Q \in R^{m\times m}$ is positive definite and $z \in \mathbb{R}^m$.
Let $z^\ast = {\rm arg} \min_{z \in \mathbbm R^m} L(z)$.
Let  $\hat{Q}$ and $\hat{r}$ be estimates of $Q$ and $r$, where each entry is independently obtained as in Lemma~\ref{lem:measurequad}. Let $B>0$ such that $\hat{Q}_{ij} \leq B$ and $\hat{r}_j \leq B$.
The solution $\hat z$ of the quadratic optimation $\hat{L}(z) = z^T \hat{Q} z - 2 \hat{r}^T z$, satisfies
\be
 L(\hat z) - L(z^\ast) \leq \epsilon.\nonumber
\ee
if a total number of measurements of $\Ord{B^2 K_A^2 m^3 \norm{Q} \norm{Q^{-1}}^2 (1+\norm{z^\ast})^2 / \epsilon}$ is used.
\end{prop}
\begin{proof}

We estimate each entry in $Q, r$ independently.
To obtain one sample on an entry of $Q, r$, we need $\Ord{K_A^2}$ quantum measurements.
Because each sample of an entry results in a value bounded by $B$,
the average of $\Ord{B^2 T}$ samples (equivalently $\Ord{B^2 K_A^2 T}$ quantum measurements) on a single entry gives a
random variable with variance $\Ord{1 / T}$.
By performing $\Ord{B^2 K_A^2 m^2 T}$ quantum measurements, we obtain an independent estimate for all the entries in $Q, r$.
Standard results in random matrix theory \cite{bandeira2016sharp} give $\norm{\hat{Q} - Q} \leq \Ord{\sqrt{m / T}}$ and $\norm{\hat{r} - r} \leq \Ord{\sqrt{m / T}}$ with high probability.

Let $z^\ast = Q^{-1} r$ be the solution to the original problem. The solution $\hat z = \hat{Q}^{-1} \hat{r}$ minimizes  $\hat{L}(z) = z^T \hat{Q} z - 2 \hat{r}^T z$.
Thus we have
$\hat{Q} \hat z - Q \hat z + Q \hat z - Q z^\ast = \hat{r} - r.$
This gives $(\hat z - z^\ast) = Q^{-1} (\hat{r} - r - (\hat{Q} - Q) \hat z)$. Hence $\norm{\hat z - z^\ast} \leq \norm{Q^{-1}}(\norm{\hat{r} - r} + \norm{\hat{Q} - Q} \norm{\hat z}) \leq \norm{Q^{-1}} \norm{\hat{r} - r} + \norm{Q^{-1}} \norm{\hat{Q} - Q} \norm{z^\ast} + \norm{Q^{-1}} \norm{\hat{Q} - Q} \norm{\hat z - z^\ast}$.
This gives
$$\norm{\hat z - z^\ast} \leq \frac{\norm{Q^{-1}}(\norm{\hat{r} - r} + \norm{\hat{Q} - Q} \norm{z^\ast})}{1 - \norm{Q^{-1}} \norm{\hat{Q}-Q}} \leq \sqrt{\epsilon / \norm{Q}}.$$
The last inequality requires setting $T > C m \norm{Q} \norm{Q^{-1}}^2 (1+\norm{z^\ast})^2 / \epsilon$, with a constant $C$, such that $1 - \norm{Q^{-1}} \norm{\hat{Q} - Q} \geq 1/2$ and $\norm{Q^{-1}}(\norm{\hat{r} - r} + \norm{\hat{Q} - Q} \norm{z^\ast}) \leq \frac{1}{2} \sqrt{\epsilon / \norm{Q}}$.
Because $L(z)$ has zero gradient at $z^\ast$, we have $$L(\hat z) - L(z^\ast) \leq \norm{Q} \norm{\hat z - z^\ast}^2 \leq \epsilon.$$
The total number of quantum measurements is
$$\Ord{B^2 K_A^2 m^2 T} = \Ord{B^2 K_A^2 m^3 \norm{Q} \norm{Q^{-1}}^2 (1+\norm{z^\ast})^2 / \epsilon}.$$
\end{proof}

The following is a detailed statement on the BQP-completeness for optimizing combination coefficients.

\begin{prop}
\label{prop:bqpcompsubspace}
Given a Hermitian matrix $A$ that have an efficient unitary decomposition and each entry can be efficiently computed classically, a quantum circuit that generates $\ket{b}$, and two quantum circuits for $\ket{u_1}, \ket{u_2}$.
It is BQP-complete to output $\hat{\alpha}_1, \hat{\alpha}_2 \in \mathbb{C}$, such that
\begin{equation}
  \label{eq:losscondition}
  L_R\left(\sum_{i=1}^2 \hat{\alpha}_i\ket{u_i} \right) \leq \min_{\alpha_1, \alpha_2 \in \mathbb{C}} L_R \left ( \sum_{i=1}^2 \alpha_i\ket{u_i} \right)  + \epsilon.
\end{equation}
\end{prop}
\begin{proof}
Consider a quantum circuit consisting of $T$ single- and two-qubit gates $U_1, \ldots, U_T$ acting on $n$ qubits.
Determining the probability of measuring $0$ in the first qubit after applying $U_1, \ldots, U_T$ on $\ket{0^n}$, $P_0 \equiv \bra{0^n} U_1^\dagger \ldots U_T^\dagger  (\ket{0}\bra{0} \otimes \mathbbm 1 \otimes \ldots \otimes \mathbbm 1) U_T \ldots U_1 \ket{0^n}$, up to small error is BQP-complete.
The same is true for the probability of measuring $1$, $P_1 = \bra{0^n} U_1^\dagger \ldots U_T^\dagger  (\ket{1}\bra{1} \otimes \mathbbm 1 \otimes \ldots \otimes \mathbbm 1) U_T \ldots U_1 \ket{0^n}$.
We now construct a linear system of dimension $2^{n+1}$, involving $1+ n$ qubits where in the notation the first qubit is now the added one.
Let the matrix $A$ be a simple controlled-NOT gate, controlled by the second qubit and acting on the first qubit.
Note that $A$ is both Hermitian and unitary, and each entry can be efficiently computed classically.
The three quantum states are given by
$$\ket{b} = \ket{0} \otimes U_T \ldots U_1 \ket{0^{n}}, \,\,\,\,
\ket{u_1} = \ket{0} \otimes U_T \ldots U_1 \ket{0^{n}}, \,\,\,\,
\ket{u_2} = \ket{1} \otimes U_T \ldots U_1 \ket{0^{n}}.$$
We can easily see that $\bra{b} A \ket{u_1} = P_0$. Similarly, $\bra{b} A \ket{u_2} = P_1$.
Suppose there is an algorithm that can efficiently find $\hat{\alpha}_1, \hat{\alpha}_2 \in \mathbb{C}$ that satisfies Equation~\eqref{eq:losscondition}.
We now show that $\hat{\alpha}_1, \hat{\alpha}_2$ can be used to infer $P_0, P_1$.
By expansion, we have
$$\Big\Vert A\Big(\sum_{i=1}^2 \alpha_i \ket{u_i}\Big) - \ket{b} \Big\Vert^2_2 = | \alpha_1 |^2 + | \alpha_2 |^2 - 2\realpar{\alpha_1 \bra{b} A \ket{u_1} + \alpha_2 \bra{b} A \ket{u_2}} + 1.$$
Using $\bra{b} A \ket{u_1} = P_0$ and $\bra{b} A \ket{u_2} = P_1$, we have
$\norm{ A\sum_{i=1}^2 \alpha_i\ket{u_i} - \ket{b} }^2_2 = |\alpha_1 - P_0|^2 + |\alpha_2 - P_1|^2 + (1 - P_0^2 - P_1^2)$.
The optimal combination parameters are $\alpha_1 = P_0$ and $\alpha_2 = P_1$.
Hence Equation~\eqref{eq:losscondition} can be rewritten as
$$|\hat{\alpha}_1 - P_0|^2 + |\hat{\alpha}_2 - P_1|^2 + (1 - P_0^2 - P_1^2) \leq (1 - P_0^2 - P_1^2) + \epsilon \implies |\hat{\alpha}_1 - P_0|^2 + |\hat{\alpha}_2 - P_1|^2 \leq \epsilon.$$
This means the algorithm can use $\hat{\alpha}_1, \hat{\alpha}_2$ to determine $P_0$ and $P_1$, which is BQP-complete.

\end{proof}

\section{Provable guarantee for the Ansatz tree approach}
\label{sec:adapproof}

We provide proof for the following propositions.
This is a simple extension and variation of known results on using polynomial approximation of $1/x$ to solve linear systems of equations \cite{childs2017quantum}.

\begin{prop}[Same as Proposition~\ref{prop:Ansatzfullproof}]
For a fixed $\epsilon \in(0,1)$, $A = \sum_{k=1}^{K_A} \beta_k U_k$ with $\rho(A) \leq 1, \rho(A^{-1}) \leq \kappa$, and $b$ with $b^\dagger b = 1$. By selecting all nodes $\{\ket{u_1}, \ldots, \ket{u_m}\}$ on the Ansatz tree with depth at most $O(\kappa \log(\kappa / \epsilon))$, we have
$$\min_{\alpha_1, \ldots, \alpha_m \in \mathbb{R}} L_R \left(\sum_i \alpha_i \ket{u_i} \right)  \leq \min_{x \in \mathbb{C}^{2^n}} L_R(x)+ \epsilon.$$
\end{prop}
\begin{proof}
By including all nodes $\{\ket{u_1}, \ldots, \ket{u_m} \}$ on the Ansatz tree with depth at most $O(\kappa \log(\kappa / \epsilon))$, the subspace contains
$p(A) b$ for any polynomial $p(\cdot)$ with degree at most $O(\kappa \log(\kappa / \epsilon))$.
In Lemma~14 in \cite{childs2017quantum}, it was shown that there exists a set of constants $p_j, \forall j = 0, \ldots, j_0$ such that $p(z) = \sum_{j = 0}^{j_0} p_j z^j$ is $\epsilon$-close to $z^{-1}$ in the domain $D_{\kappa} = [-1, -1/\kappa] \cup [1/\kappa, 1]$, where $j_0 = 2 \sqrt{\kappa^2 \log(2 \kappa / \epsilon) \log(8 \kappa^2 \log(2 \kappa / \epsilon) / \epsilon)} + 1 = O(\kappa \log(\kappa / \epsilon))$.
Using the condition that $\rho(A) \leq 1$ and $\rho(A^{-1}) \leq \kappa$, we know that all the eigenvalues of $A$ lie in the domain $D_{\kappa}$.
Because $p(z)$ is $\epsilon$-close to $z^{-1}$ in the domain $D_{\kappa}$, we thus have $\norm{p(A) - A^{-1}} \leq \epsilon$.
This implies that $\norm{p(A)b - A^{-1}b} \leq \norm{p(A) - A^{-1}} \leq \epsilon$.
Hence there exists a set of combination parameters $\hat{\alpha}_1, \ldots, \hat{\alpha}_m \in \mathbb{R}$ set according to the coefficients $p_j$ in the polynomial $p(x)$, such that $\hat{x} = \sum_i \hat{\alpha}_i \ket{u_i}$ satisfies $\norm{\hat{x} - A^{-1}b}_2 \leq \epsilon$.
So
$$\min_{\alpha_1, \ldots, \alpha_m \in \mathbb{R}} \Big\Vert A \Big(\sum_i \alpha_i \ket{u_i}\Big) - b \Big\Vert_2^2 \leq \norm{A\hat{x} - b}_2^2 \leq \rho(A)^2 \norm{\hat{x} - A^{-1}b}_2^2 \leq \epsilon^2 = \min_{x \in \mathbb{C}^{2^n}} \norm{Ax - b}_2^2 + \epsilon^2.$$
The last equality uses the fact that $x = A^{-1} b$ satisfy $\norm{Ax - b}_2 = 0$.
Note that we actually achieve $\epsilon^2$ error, which is better than $\epsilon$ since $\epsilon < 1$.
\end{proof}

\begin{prop}[Same as Proposition~\ref{prop:Ansatzproof}]
For a fixed $\epsilon \in (0,1)$, and $A = \sum_{k=1}^{K_A} \beta_k U_k$ with $\rho(A) \leq 1$. By selecting all nodes $\{\ket{u_1}, \ldots, \ket{u_m} \}$ on the Ansatz tree with depth at most $\ceil{\log(1 / 2 \epsilon) / \log(1 / (2 - \sqrt{3}))}$, we have
$$\min_{\alpha_1, \ldots, \alpha_m \in \mathbb{R}} L_T\left(\sum_i \alpha_i \ket{ u_i}\right) \leq \min_{x \in \mathbb{C}^{2^n}}L_T(x) + \epsilon.$$
\end{prop}
\begin{proof}
We first diagonalize $A$ to be $V D V^\dagger$, where $D$ is diagonal and $V$ is a unitary matrix.
We also set $N = 2^n$ to be the system size.
The eigenvalues of $A$ are denoted as $\lambda_i, \forall i = 1, \ldots, N$.
We set $\tilde{b} = V^\dagger \ket b$ and note that $\sum_{i=1}^{N} |\tilde{b}_i|^2 = 1$, as $\ket b$ is normalized.
We consider a rotated $x$, $\tilde{x} = V^\dagger x \in \mathbb{C}^{N}$.
Using $\tilde{x}$, the loss function $\frac{1}{2} \norm{ x}_2^2 + \norm{ Ax - \ket b }_2^2$ can be written as
$$\sum_{i=1}^N \Bigg(\frac{1}{2} |\tilde{x}_i|^2 + |\lambda_i \tilde{x}_i - \tilde{b}_i|^2 \Bigg).$$
We can minimize this expression analytically as
$$\tilde{x}_i = \frac{2\lambda_i \tilde{b}_i}{2 \lambda_i^2 + 1} \in \mathbb{C}, \forall i = 1, \ldots, N.$$
Plugging this optimal solution into the loss function yields
$$\sum_{i=1}^N |\tilde{b}_i|^2 \Bigg(\frac{1}{2} y_i^2  + (\lambda_i y_i - 1)^2 \Bigg),$$
where $y_i = 2\lambda_i / (2 \lambda_i^2 + 1) \in \mathbb{R}$.
We now consider the space of all linear combinations of $A^k \ket b, \forall k = 0, \ldots, K_0$, which is a subspace of $\mbox{span}(u_1, \ldots, u_m)$.
This space is written as $\Big\{\sum_{k=0}^{K_0} p_k A^k \ket b\Big\}$.
In this subspace, the loss function can be written as
$$\frac{1}{2} \norm{\sum_k p_k D^k \tilde{b}}_2^2 + \norm{\sum_k p_k D^{k+1} \tilde{b} - \tilde{b}}_2^2
= \sum_{i = 1}^N |\tilde{b}_i|^2 \Bigg(\frac{1}{2} \Big(\sum_k p_k \lambda_i^k\Big)^2 + \Big(\lambda_i \sum_k p_k \lambda_i^k - 1\Big)^2 \Bigg).$$
We now analyze how accurate a polynomial $\sum_k p_k x^k$ can approximate $2x / (2 x^2 + 1)$ within $[-1, 1]$.
Due to the condition that $\rho(A) \leq 1$, we only care about the range $[-1, 1]$.
The approximation can be done by performing Chebyshev decomposition of the function $x / (x^2 + 1/2)$,
$$\frac{x}{x^2 + 1 / 2} = \sum_{k=0, 1, 2, \ldots} c_k T_{2k + 1}(x),$$
where $T_{2k+1}(x)$ is $(2k+1)$-th Chebyshev polynomial of the first kind (which is of degree $2k+1$). And we have the following recursive formula for $c_k, \forall k = 0, 1, 2, \ldots$,
$$c_k = (-2 + \sqrt{3})^{k} \Big(1 - \frac{1}{\sqrt{3}}\Big).$$
If we truncate the Chebyshev expansion at $\lfloor (K_0 - 1) / 2 \rfloor$ (the degree is at most $K_0$), then
$$\sup_{x \in [-1, 1]} \Bigg|\frac{x}{x^2 + 1/2} - \sum_{k=0}^{\lfloor (K_0 - 1) / 2 \rfloor} c_k T_{2k+1}(x) \Bigg| \leq \sum_{\lfloor (K_0 - 1) / 2 \rfloor + 1}^{\infty} |c_k| \leq \frac{(2 - \sqrt{3})^{K_0 / 2} \Big(1 - \frac{1}{\sqrt{3}}\Big)}{\sqrt{3}-1} \equiv \eta.$$
By choosing $p_k$ according to $c_k$ and the Chebyshev polynomial coefficients, we have
$$\Big| y_i - \sum_{k=0}^{K_0} p_k \lambda_i^k \Big| \leq \eta, \forall i = 1, \ldots, N.$$
Then using $\frac{1}{2} z^2 + (\lambda_i z - 1)^2 = \frac{1}{2} y_i^2 + (\lambda_i y_i - 1)^2 + (\frac{1}{2} + \lambda_i^2) (z - y_i)^2, \forall z \in \mathbb{R}$,
we have
$$\frac{1}{2} \Big(\sum_k p_k \lambda_i^k\Big)^2 + \Big(\lambda_i \sum_k p_k \lambda_i^k - 1\Big)^2 \leq \frac{1}{2} y_i^2  + (\lambda_i y_i - 1)^2 + \Big(\frac{1}{2} + \lambda_i^2\Big) \eta^2, \forall i = 1, \ldots, N.$$
Using the fact that $|\lambda_i| \leq 1$ and $\sum_{i} |\tilde{b}_i|^2 = 1$, we have
$$\sum_{i = 1}^N |\tilde{b}_i|^2 \Bigg(\frac{1}{2} \Big(\sum_k p_k \lambda_i^k\Big)^2 + \Big(\lambda_i \sum_k p_k \lambda_i^k - 1\Big)^2 \Bigg) \leq \min_{x \in \mathbb{C}^{2^n}} \Bigg( \frac{1}{2} \norm{ x}_2^2 + \norm{ Ax - b }_2^2 \Bigg) + \frac{3}{2} \eta^2.$$
Now we want $\frac{3}{2}\eta^2 \leq \epsilon$ by choosing a large enough $K_0$.
Using $\eta = \frac{(2 - \sqrt{3})^{K_0 / 2}}{\sqrt{3}}$, we need $(2 - \sqrt{3})^{K_0} \leq 2 \epsilon$.
By choosing $K_0\geq \log(1 / 2\epsilon) / \log(1 / (2 - \sqrt{3}))$, we are guaranteed to have
$$\min_{\alpha_1, \ldots, \alpha_m \in \mathbb{R}} \Big( \frac{1}{2} \Big\Vert \sum_i \alpha_i u_i \Big\Vert_2^2 + \Big\Vert A \Big(\sum_i \alpha_i u_i \Big) - b \Big\Vert_2^2 \Big) \leq \min_{x \in \mathbb{C}^{2^n}} \Big( \frac{1}{2} \norm{x}_2^2 + \norm{ Ax - b }_2^2 \Big) + \epsilon.$$
\end{proof}

\begin{prop}[Same as Proposition~\ref{prop:treegradient}]
\label{prop:treegradientappendix}
Consider a subspace $S$ with $m$ states $\ket{\psi_1}, \ldots, \ket{\psi_m}$, and the optimal $x^S = \sum_i \alpha^*_i \ket{\psi_i}$.
If the gradient overlap of $\ket{\psi^*}$ is
$$g = \left|\bra{\psi^*} \nabla L_R(x^S)\right| = \left|2 \sum_{\ket{\psi_i} \in S} \alpha^*_{i} \bra{\psi^*} A^2 \ket{\psi_i} - 2 \bra{\psi^*} A \ket b \right|,$$
then the loss function will have a guaranteed decrease in the next round given by
$$\min_{\alpha_1, \ldots, \alpha_{m+1} \in \mathbb{C}}  L_R \left(\sum_{\ket{\psi_i} \in S \cup \{\ket{\psi^*}\}} \alpha_i \ket{\psi_i}\right) \leq \min_{\alpha_1, \ldots, \alpha_m  \in \mathbb{C}}  L_R\left(\sum_{\ket{\psi_i} \in S} \alpha_i \ket{\psi_i}\right) - \frac{g^2}{4}.$$
\end{prop}
\begin{proof}
Consider the loss function on $x^S + \alpha \ket{\psi^*}$, where $\alpha \in \mathbb{C}$.
Because $x^S$ is a linear combination of $\ket{\psi_i} \in S$ and $L_R(x) = \norm{Ax - b}_2^2$, we know that the optimal combination of $S \cup \{\ket{\psi^*}\}$ satisfies
$$\min_{\alpha_1, \ldots, \alpha_{m+1} \in \mathbb{C}}  L_R \left(\sum_{\ket{\psi_i} \in S \cup \{\ket{\psi^*}\}} \alpha_i \ket{\psi_i}\right) \leq L_R(x^S + \alpha \ket{\psi^*}),$$
for any $\alpha \in \mathbb{C}$.
By expanding $L_R(x^S + \alpha \ket{\psi^*})$, we have ($\overline \alpha$ denotes the complex conjugated $\alpha$)
$$L_R(x^S + \alpha \ket{\psi^*}) = L_R(x^S) + |\alpha|^2 \bra{\psi^*} A^\dagger A \ket{\psi^*} + \realpar{\overline \alpha \bra{\psi^*} \nabla L_R(x^S)}.$$
By selecting $\alpha = - \overline{\bra{\psi^*} \nabla L_R(x^S)} / 2 \bra{\psi^*} A^\dagger A \ket{\psi^*}$, we have
$$L_R(x^S + \alpha \ket{\psi^*}) = L_R(x^S) - \frac{|\bra{\psi^*} \nabla L_R(x^S)|^2}{4 \bra{\psi^*} A^\dagger A \ket{\psi^*}} \leq \min_{\alpha_1, \ldots, \alpha_m  \in \mathbb{C}}  L_R\left(\sum_{\ket{\psi_i} \in S} \alpha_i \ket{\psi_i}\right)  - \frac{g^2}{4}.$$
The last inequality uses the assumption that the spectral radius of $A$ is no greater than $1$.
\end{proof}

\section{Experiments on solving linear systems using Agnostic Ans\"atze}
\label{appAgnostic}

\begin{figure}[t]
\centering
\includegraphics[width=1.0\textwidth]{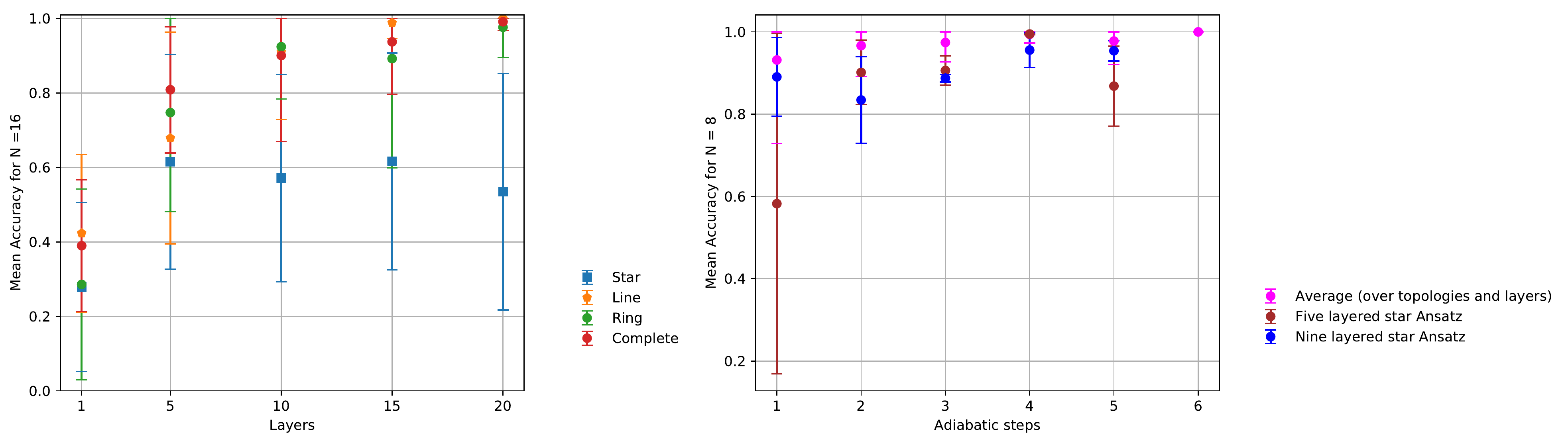}
\caption{Numerical experiments on solving linear systems using Agnostic Ans\"atze.
Left: The mean accuracy versus circuit layer depth for various Ans\"atze employed to solve linear systems for $N=16$ is shown. The mean accuracy goes to unity for line graph, ring graph and complete graph Ans\"atze with a $20$-layer circuit. The star-graph Ansatz performs performs worst and eventually levels at accuracy $0.6$.
Right: We implement the adiabatic-assisted VQE (AAVQE) approach for $N=8$ and see improvement over standard VQE (adiabatic steps $= 1$) over an average of all Ans\"atze and number of layers. Employing more adiabatic steps also improves the mean accuracy. For example, the mean accuracy for the five and nine layered  star Ans\"atze improves to 1.0 as we increase the number of adiabatic steps from $1$ to $6$.
The dots represent the mean accuracy and the bars represent the spread corresponding to a standard deviation with the upper and lower cutoffs as $1$ and $0$ respectively.}
\label{fig:comp}
\end{figure}

We discuss our experiments on solving linear systems using the Agnostic Ansatz in detail.
We focus on solving the real-valued version of Eq.~\eqref{linear}. By construction, our Agnostic Ansatz is constrained to explore the solution vector in the real subspace of the appropriate Hilbert space.
The two real gates we use are the single-qubit rotation around the $y$-axis for every qubit with tunable angle (the variational parameter) and the controlled NOT (CNOT) gate. Thus, a single layer of our $n$-qubit variational circuit consists of $n$ variational parameters and a certain pattern of CNOT gates.
We implement various types of Agnostic Ans\"atze depending on how the CNOT gates are applied. The topology of a quantum computer favours a particular arrangement of CNOT gates over another and our exploration for different Ans\"atze is motivated by the same. Here, we enumerate the different variational Ans\"atze. We label our qubits from $1$ to $n$ and denote the same by $[1, \cdots,  n]$. The CNOT gate between qubit $i$ (control) and $j$ (target) will be denoted by $C(i,j)$.
\begin{enumerate}
\item Star Ansatz: The qubit numbered $1$ is always the control, while target ranges over all $i \in [2, \cdots, n]$. In other words, we apply $C(1,i)$ for all $i \in [2,\cdots,n]$.
\item Line Ansatz: The Ansatz contains $C(i,i+1) $ for every $ i \in [1,\cdots,n-1]$.
\item Ring Ansatz:  It is similar to the line Ansatz with the difference that there is an extra CNOT gate at the boundary, i.e., $C(n,1)$.
\item Complete graph Ansatz: We implement $C(i,j)$ for every $ i,j  \in [1,\cdots,n]$ such that $i \neq j$.
\end{enumerate}

We have conducted numerical experiments on Rigetti quantum virtual machine where the linear systems are generated randomly over different system sizes ($N = 2, 4, 8, 16$).
Some of the numerical results can be seen in Figure~\ref{fig:comp}.
The figure of merit is mean accuracy, which is the average fidelity of the output vector and the solution over $100$ independent runs.
In Figure~\ref{fig:comp} (Left), we present the use of standard VQE for solving linear systems with $N=16$.
We can see a rise in overall performance as we increase the number of layers. The mean accuracy goes to unity for most CNOT gate patterns except for the star graph.
The performance for the star graph starts improving but soon levels at mean accuracy of $0.6$.

In Figure~\ref{fig:comp} (Right), we show the adiabatic-assisted VQE approach for $N=8$.
The purple/magenta data points are an average over all the topologies and different layers.
We can see an improvement by roughly $10\%$ using AAVQE (adiabatic steps $ = 6$) over standard VQE (adiabatic steps $ = 1$).
We can also see that as we increase the number of adiabatic steps, the performance of AAVQE becomes better.
The mean accuracy for all settings considered here becomes very close to unity as we increase the number of adiabatic steps to $6$. Furthermore, the standard deviation around the mean accuracy goes below $10^{-2}$.
As such, most of the settings were able to achieve the accuracy of close to $1.0$ in AAVQE, which is not achieved with standard VQE (adiabatic steps $ = 1$).
\end{document}